\documentclass[a4paper,11pt,reqno]{amsart}
\usepackage[left=25mm,top=25mm,bottom=25mm,right=25mm]{geometry}
\usepackage{dsfont,amscd}
\usepackage[cyr]{aeguill}
\usepackage[pdfdisplaydoctitle=true,
            colorlinks=true,
            urlcolor=blue,
            citecolor=blue,
            linkcolor=blue,
            pdfstartview=FitH,
            pdfpagemode=None,
            bookmarksnumbered=true]{hyperref}
\usepackage{graphicx}
\usepackage{todonotes}
\usepackage{amssymb}
\usepackage{enumerate}
\usepackage{pstricks,pstricks-add,pst-math}
\usepackage{subfigure}
\newtheorem{thm}{Theorem}[section]
\newtheorem{cor}[thm]{Corollary}
\newtheorem{lem}[thm]{Lemma}

\newtheorem{prop}[thm]{Proposition}
\theoremstyle{definition}
\newtheorem{defn}[thm]{Definition}
\theoremstyle{remark}
\newtheorem{rem}[thm]{Remark}
\numberwithin{equation}{section}

\newtheorem*{thmI}{\textbf{\emph{Theorem I}}}
\newtheorem*{thmII}{\textbf{\emph{Theorem II}}}
\newtheorem*{thmIII}{\textbf{\emph{Theorem III}}}

\newcommand{\RR}{\mathbb R}

\newcommand{\GR}{\mathbb{R}}
\newcommand{\TT}{\mathcal{T}}
\newcommand{\OO}{\mathcal{O}}
\newcommand{\II}{\mathcal{I}}
\newcommand{\CC}{\mathrm{Z}}
\newcommand{\DD}{\mathrm{D}}

\newcommand{\sot}{\mathrm{SO}(3)}
\newcommand{\sud}{\mathrm{SU}(2)}
\newcommand{\ode}{\mathrm{O}(2)}
\newcommand{\otr}{\mathrm{O}(3)}
\newcommand{\sod}{\mathrm{SO}(2)}
\newcommand{\id}{\mathds{1}}

\newcommand{\xx}{\mathbf{x}}

\newcommand{\ve}{\mathbf{v}}
\newcommand{\yy}{\mathbf{y}}

\newcommand{\ii}{\mathbf{i}}
\newcommand{\jj}{\mathbf{j}}
\newcommand{\kk}{\mathbf{k}}
\newcommand{\QQ}{\mathbf{Q}}

\newcommand{\rr}{\mathbf{r}}

\newcommand{\GL}{\mathrm{GL}}     
\newcommand{\Ela}{\mathbb{E}\mathrm{la}}    

\newcommand{\beq}{\begin{equation}}
\newcommand{\eeq}{\end{equation}}
\newcommand{\ben}{\begin{equation*}}
\newcommand{\een}{\end{equation*}}
\newcommand{\ba}{\begin{eqnarray}}
\newcommand{\ea}{\end{eqnarray}}
\newcommand{\ban}{\begin{eqnarray*}}
\newcommand{\ean}{\end{eqnarray*}}

\newcommand{\Listerd}{$\bullet$}

\newenvironment{listerd}%
	{\begin{list}{\Listerd}{}}%
	{\end{list}}

\hypersetup{
    pdftitle={Symmetry classes for even-order tensors}, 
    pdfauthor={M. Olive and N. Auffray}, 
    pdfsubject={MSC: 20C35,74B99, 15A72}, 
    pdfkeywords={Anisotropy, Symmetry classes, Higher order elasticity}, 
    pdflang=EN 
    }

\begin{document}

\title{Symmetry classes for even-order tensors}%

\author{M. Olive}
\address{LATP, CNRS \& Universit\'{e} de Provence, 39 Rue F. Joliot-Curie, 13453 Marseille Cedex 13, France}
\email{molive@cmi.univ-mrs.fr}

\author{N. Auffray}
\address{LMSME, Universit\'{e} Paris-Est, Laboratoire Mod\'{e}lisation et Simulation Multi Echelle,MSME UMR 8208 CNRS, 5 bd Descartes, 77454 Marne-la-Vall\'{e}e, France}
\email{Nicolas.auffray@univ-mlv.fr}


\date{\today}%
\begin{abstract}

The purpose of this article is to give a complete and general answer to the recurrent problem in continuum mechanics of the determination of the number and the type of symmetry classes of an even-order tensor space. This kind of investigation was initiated for the space of elasticity tensors. Since then, different authors solved this problem for other kinds of physics such as photoelectricity, piezoelectricity, flexoelectricity, and strain-gradient elasticity. All the aforementioned problems were treated by the same computational method. Although being effective, this method suffers the drawback not to provide general results. And, furthermore, its complexity increases with the tensorial order. In the present contribution, we provide general theorems that directly give the sought results for any even-order constitutive tensor. As an illustration of this method, and for the first time, the symmetry classes of all even-order tensors of Mindlin second strain-gradient elasticity are provided.
\end{abstract}

\maketitle

%

\section{Introduction}

\subsection{Physical motivation}

Within these last years, an increasing interest in generalized continuum theories \cite{For98,DSV09,LS11,DSM12} has been observed. These works based on the pioneering works of \cite{Min64,Min65,Tou62} propose extended kinematic formulation in order to take into account size effects within the continuum. The price to pay is the appearance in the constitutive relations of tensors of order greater than $4$. These higher order objects are difficult to handle and extracting physical meaningful information is not straightforward.
The aim of this paper is to provide general results concerning the type and number of anisotropic systems an even-order tensor can have.\\

\textbf{Such results have important applications, at least, for the modeling and the numerical implementation of non classical linear constitutive laws:}
\begin{description}
\item[Modeling]\textbf{ The stakes of modeling is given a material, and a set of physical variables of interest to construct the more general (linear, in the present context) constitute law that describes the behavior of that material. An example of such a method is provided by Thionnet and Martin \cite{TM06} where, given a set of variables $V$ and the material symmetry group $S$, they derive mechanical behavior laws using the Theory of Invariants and Continuum Thermodynamics.
In such perspective our results will say, without making any computation whether or not $S$ is contained in the set of symmetry classes of $\mathcal{L}(v,v')$ the space of linear applications from $v\in V$ to $v'\in V$.}

\item[Numerical implementation]\textbf{ To implement a new linear constitutive law in a finite-element code, one has to known the complete set of matrices needed to model the associated anisotropic behavior. In that perspective, our result in a precious guideline as it tells you the number of the sought matrices, and the way to construct them. Such a way to proceed is illustrated, for example, in the case of $3$D strain gradient elasticity in \cite{ALH12}.}
\end{description}

\subsubsection*{Constitutive tensors symmetry classes}

In the field of mechanics, constitutive laws are usually expressed in terms of tensorial relations between the gradient of primary variables and their fluxes \cite{GH11}. As it is well-known, this feature is not restricted to linear behaviors since tensorial relations appear in the tangential formulation of non-linear ones \cite{TB96}.
It is also known that a general tensorial relation can be divided into classes according to its symmetry properties. Such classes are referred  to as symmetry classes in the field of mechanics \cite{FV96}, or as isotropic classes (or strata) in the field of mathematical physics \cite{AS83,AKP12}.

In the case of second order tensors, the determination of symmetry classes is rather simple. According to a spectral analysis it can be concluded that any second-order symmetric tensor\footnote{Such a tensor is related to a symmetric matrix, which can be diagonalize in an orthogonal basis. The  stated result is related to this diagonalization.} can either be orthotropic ($[\DD_{2}]$), transverse isotropic ($[\ode]$), or isotropic ($[\sot]$). Such kind of tensors are commonly used to describe, e.g., heat conduction, electric permittivity.
For higher order tensors, the determination of the set of symmetry classes is more involved, and is mostly based on an approach introduced by Forte and Vianello \cite{FV96} in the case of elasticity. Let us briefly detail this case.\\

The vector space of elasticity tensors, which will be noted $\Ela$ throughout this paper, is the subspace of 4th-order tensors endowed with the following index symmetries:
\begin{description}
\item[Minor symmetries]: $E_{ijkl}=E_{jikl}=E_{jilk}$
\item [Major symmetry]: $E_{ijkl}=E_{klij}$
\end{description}
The symmetries will be denoted using the following notation: $E_{\underline{(ij)}\ \underline{(kl)}}$, where $(..)$ indicates the invariance under the in-parenthesis indices permutations, and  $\underline{..}\ \underline{..}$ the invariance with respect to the underlined blocks permutations.
Index symmetries encode the physics described by the mathematical operator. \textbf{On one hand, the minor symmetries stem from the fact that rigid body motions do not induce deformation (symmetry of $\varepsilon$), and that the material is not subjected to volumic couple (symmetry of $\sigma$). And, on the other hand, the major symmetry is the consequence of the existence of a free energy.}
An elasticity tensor, $\mathbf{E}$, can be viewed as a symmetric linear operator on $\mathbb{T}_{(ij)}$, the space of symmetric second order tensors.
According to Forte and Vianello \cite{FV96}, for the classical action of $\sot$, $\Ela$ is divided into the following $8$ symmetry classes:
\ben
[\Ela]=\{[\id],[\CC_2],[\DD_2], [\DD_3], [\DD_4], [\ode] ,[\OO],[\sot]\}
\een 
which correspond, respectively, to the following physical classes\footnote{These symmetry classes are subgroups of the spatial group of rotations  $\sot$. This is due to the fact that the elasticity tensor is even-order. To treat odd-order tensors, the full orthogonal group $\otr$ has to be considered.} : triclinic ($[\id]$), monoclinic ($[\CC_2]$), orthotropic ($[\DD_2]$), trigonal ($[\DD_3]$), tetragonal ($[\DD_4]$), transverse isotropic ($[\ode]$), cubic ($[\OO]$) and isotropic ($[\sot]$). The mathematical notations used  above will be detailed in \autoref{ss:OtrSub}. 
Besides this fundamental result, the interest of the Forte and Vianello paper was to provide a general method to determine the symmetry classes of any tensor space \cite{AKP12}. Since then, using this method, other results have been obtained:
\ben
\begin{tabular}{|c|c|c|c|c|}
  \hline
  Property&Tensor&Number of classes&Action&Studied in\\ \hline\hline
  Photoelelasticity & $T_{(ij)(kl)}$ & 12 &$\sot$ &  \cite{FV97} \\ \hline
  Piezoelectricity & $T_{(ij)k}$& 15&$\otr$&\cite{GW02}\\ \hline
  Flexoelectricity& $T_{(ij)kl}$& 12&$\sot$&\cite{LH11}\\ \hline
  A set of $6$-th order tensors&$T_{ijklmn}$&14 or 17&$\sot$&\cite{LAH+12} \\
  \hline
\end{tabular}
\een


\subsubsection*{The limitations of the Forte-Vianello approach}
The method introduced by Forte and Vianello is actually the more general one\footnote{Some other methods can be found in the literature, such as counting the symmetry planes\cite{CVC01}, or studying the $\sud$-action on $\Ela$ \cite{BBS04}, $\ldots$ but these methods are difficult to generalize to any kind of vector space.}. But, in the same time, it suffers from at least two limitations:
\begin{enumerate}
\item The computation of the harmonic decomposition;
\item The specificity of the study for each kind of tensor.
\end{enumerate}

In its original setting, the method requires the computation of the explicit harmonic decomposition of the studied tensor, i.e. its decomposition into the sum of its $\sot$-irreducible components, also known as harmonic tensors\footnote{Harmonic tensors are completely symmetric and traceless, they inherit this name because of a well-known isomorphism in $\RR^{3}$ between these tensors and harmonic polynomials \cite{Bac70}.}. Its explicit computation, which is generally based on an algorithm introduced by Spencer \cite{Spe70}, turns out to be intractable in practice as the tensorial order increases. But, this is not a real problem, since the only information needed is the number of different harmonic tensors of each order appearing in the decomposition, i.e. the isotypic decomposition. Based on arguments presented in \cite{JCB78}, there exists a direct procedure to obtain this isotypic decomposition from the tensor index symmetries \cite{Auf08}. Such an approach has been used in \cite{LAH+12} to obtain the symmetry classes of 6th-order tensors.\\

As each kind of tensor space requires a specific study, this specificity constitutes the  other limitation of the method. This remark has to be considered together with the observation that, for even-order tensors it seems that there exists, indeed, only two possibilities. Precisely, a tensor space has as many classes as:
\begin{itemize}
\item the full symmetric tensor space: e.g. $\Ela$ is divided into 8 classes as the full symmetric one \cite{FV96};
\item the generic tensor space\footnote{The $n$th-order generic tensor is a $n$th-order tensor with no index symmetry.}: any other 4th-order tensor space (photoelasticity \cite{FV97}, flexoelectricity \cite{LH11},...) is divided into $12$ classes such as the generic one.
\end{itemize}
The same observation can also be made for $2$nd- and $6$th-order tensors \cite{LAH+12}. Understanding what is the general rule behind this observation would be an important result in mechanics. Its practical implication is the direct determination of the number and the type of symmetry classes for any constitutive laws no matter their orders.
This result is of valuable importance to understand the feature of generalized continuum theories, in which higher order tensors are involved in constitutive laws.

\subsection{Organization of the paper}
In \autoref{s:MaiRes}, the mains results of this paper: \textbf{Theorems I, II} and \textbf{III} are stated. As an application, the symmetry classes of the even-order constitutive tensor spaces of Mindlin second strain-gradient elasticity are determined. Results concerning the  $6$th-order coupling tensor and the  $8$th-order second strain-gradient tensor are given for the first time. 
Obtaining the same results with the Forte-Vianello approach would have been much more difficult.
Other sections are dedicated to the construction of our proofs. In \autoref{s:GenFra}, the mathematical framework used to obtain our result is introduced. Thereafter, we study the symmetry classes of a couple of harmonic tensors, which is the main purpose of the tool named \textit{clips operator}. We then give the associated results for couples of $\sot$-closed subgroups (theorem \ref{thm:clipssot} and table \ref{global}). Thanks to these results, and with the help of some previous works on that topic done by Golubistky and al. \cite{GSI84}, we obtain in \autoref{lastone} some general results concerning symmetry classes for general even-order tensors. In \autoref{IsoConst} our main results are finally proved. The appendix is devoted to proofs and calculus of clips operations. 

\section{Main results}\label{s:MaiRes}

In this section, our main results are stated.
In the first subsection, the construction of \emph{Constitutive Tensor Spaces} (CTS in the following) is discussed. This construction allows us to formulate our main results in the next subsection. Finally, application of these results to Mindlin Second Strain-Gradient Elasticity (SSGE in the following) is considered. It worth noting that precise mathematical definitions of the symmetry classes are given in \autoref{s:GenFra}. 


\subsection{Construction of CTS}

Linear constitutive laws are linear applications between the gradient of primary physical quantities and theirs fluxes. Each of these physical quantities (see table~\ref{tab:CTSStructure}) are in fact related to subspaces\footnote{because of some symmetries} of tensors spaces: theses subspaces will be called \textit{State Tensor Spaces} (STS in the following). These STS will be the primitive notion from which the CTS will be constructed. \textbf{In the following, $\mathcal{L}(F,G)$ will indicate the vector space of linear application from $F$ to $G$.}



\begin{table}[h]
\begin{tabular}{|c|c|c|}
\hline
Physical notions & Mathematical object & Mathematical space \\ \hline \hline
Gradient & Tensor state $\mathbf{T}_1\in \otimes^p\RR^3$ & $\mathbb{T}_G$: tensor space with index symmetries \\ 
\hline
Fluxes of gradient & Tensor state $\mathbf{T}_2\in \otimes^q\RR^3$ & $\mathbb{T}_f$: tensor space with index symmetries \\
\hline
Linear constitutive law & $C\in \mathcal{L}(\mathbb{T}_G,\mathbb{T}_f)$ & $\mathbb{T}_C\subset \mathcal{L}(\mathbb{T}_G,\mathbb{T}_f)$ \\
\hline
\end{tabular}\caption{Physical and mathematical links}\label{tab:CTSStructure}
\end{table}

Let consider now two STS : $\mathbb{E}_{1}=\mathbb{T}_G$ and $\mathbb{E}_{2}=\mathbb{T}_f$, which are $p$th- and $q$th-order STS with possibly index symmetries. As a consequence, they belong to subspaces of $\otimes^{p}\RR^{3}$ and  $\otimes^{q}\RR^{3}$.
\textit{A constitutive tensor} $C$ is a linear application between $\mathbb{E}_{1}$ and $\mathbb{E}_{2}$, i.e. an element of the space $\mathcal{L}(\mathbb{E}_{1},\mathbb{E}_{2})$. This space is isomorphic, modulo the use of an euclidean metric, to $\mathbb{E}_{1}\otimes \mathbb{E}_{2}$. Physical properties leads to some index symmetries on $\mathbf{T}\in \mathbb{E}_1\otimes \mathbb{E}_2$; thus the vector space of such $\mathbf{T}$ is some vector subspace $\mathbb{T}_C$ of $\mathbb{E}_1\otimes \mathbb{E}_2$.
  
Now, each of the spaces $\mathbb{E}_1$, $\mathbb{E}_2$ and $\mathbb{E}_1\otimes \mathbb{E}_2$ has a natural $\otr$ action. In this paper, we are concerned with cases in which $p+q=2n$. In such a situation, it is known that the $\otr$-action on $\mathbb{E}_{1}\otimes \mathbb{E}_{2}$ reduces to the one of $\sot$ \cite{FV96}.
\noindent We therefore have
\begin{equation*}
\mathcal{L}(\mathbb{E}_{1},\mathbb{E}_{2})\simeq \mathbb{E}_1\otimes \mathbb{E}_2\subset \mathbb{T}^{p}\otimes\mathbb{T}^{q}=\mathbb{T}^{p+q=2n}
\end{equation*}
Examples of such constructions are provided in the following table:
\ben
\begin{tabular}{|c|c|c|c|c|}
  \hline
  Property&$\mathbb{E}_{1}$&$\mathbb{E}_{2}$&Tensor product for CTS& Number of classes \\ \hline\hline
  Elasticity & $\mathbb{T}_{(ij)}$&$\mathbb{T}_{(ij)}$ & Symmetric & 8 \\ \hline
  Photoelelasticity & $\mathbb{T}_{(ij)}$&$\mathbb{T}_{(ij)}$ & Standard & 12 \\ \hline
  Flexoelectricity& $\mathbb{T}_{(ij)k}$&$\mathbb{T}_{i}$ & Standard & 12 \\ \hline
  First-gradient elasticity& $\mathbb{T}_{\underline{(ij)k}}$ &$\mathbb{T}_{\underline{(ij)k}}$& Symmetric & 17 \\ \hline
\end{tabular}
\een
In the former table, two kinds of constitutive tensor spaces appeared whether they describe
\begin{itemize}
\item Coupled physics : such tensors encode the coupling between two different physics, such as photoelasticity and flexoelectricity;
\item Proper physics  : such tensors describe solely one physical phenomenon, such as classical and first-gradient elasticities.
\end{itemize}
On a mathematical side this implies :
\begin{itemize}
\item Coupled physics : the spaces $\mathbb{E}_{1}$ and $\mathbb{E}_{2}$ may differ, and when $\mathbb{E}_{1}=\mathbb{E}_{2}$ linear applications are not self-adjoint;
\item Proper physics  : we have $\mathbb{E}_{1}=\mathbb{E}_{2}$ and linear applications are self-adjoint\footnote{This is a consequence of the assumption of the existence of a free energy.}.
\end{itemize}
Therefore, the \textit{elasticity tensor} is a self-adjoint linear application between the vector space of deformation tensors and the vector space of stress tensors. These two spaces are modeled on $\mathbb{T}_{(ij)}$. The vector space of elasticity tensors is therefore completely determined by $\mathbb{T}_{(ij)}$ and the symmetric nature of the tensor product, i.e. $\Ela=\mathbb{T}_{(ij)}\otimes^{S}\mathbb{T}_{(kl)}$ where $\otimes^{S}$ denotes the symmetric tensor product.
On the side of coupling tensors, the flexoelectricity is a linear application between $\mathbb{E}_{1}=\mathbb{T}_{(ij)k}$, the space of deformation gradient, and $\mathbb{E}_{2}=\mathbb{T}_{l}$ the electric polarization, therefore $\mathbb{F}\mathrm{lex}=\mathbb{T}_{(ij)k}\otimes\mathbb{T}_{l}$.

\subsection{Symmetry classes of even order tensor spaces}\label{ss:OtrSub}

Let consider an even order constitutive tensor space $\mathbb{T}^{2n}$, it is known \cite{JCB78} that this space can be decomposed orthogonally\footnote{The related dot product is constructed by $2n$ products of the $\mathbb{R}^{3}$ canonical one.} into a full symmetric space and a complementary one which is isomorphic to a tensor space of order $2n-1$, i.e. :
\begin{equation*}
\mathbb{T}^{2n}=\mathbb{S}^{2n}\oplus \mathbb{C}^{2n-1}
\end{equation*}
Let us introduce :
\begin{description}
\item [$\mathbb{S}^{2n}$] the vector space of $2n$th-order completely symmetric tensors;
\item [$\mathbb{G}^{2n}$] the vector space of $2n$th-order tensors with no index symmetries\footnote{Formally this space is constructed as $\mathbb{G}^{2n}=\otimes^{2n}\mathbb{R}^3$.}.
\end{description}
The following observation is obvious :
\ben
\mathbb{S}^{2n}\subseteq\mathbb{T}^{2n}\subseteq\mathbb{G}^{2n}
\een
and therefore, if we note $\mathfrak{I}$ the operator which to a tensor space gives the set of its symmetry classes, we obtain:
\ben
\mathfrak{I}(\mathbb{S}^{2n})\subseteq\mathfrak{I}(\mathbb{T}^{2n})\subseteq\mathfrak{I}(\mathbb{G}^{2n})
\een
Symmetry group of even-order tensors are conjugate to $\sot$-closed subgroups \cite{ZB94, FV96}. The classification of $\sot$-closed subgroups is a classical result that can be found in many references \cite{GSI84,Ste94}. These subgroups are, up to conjugacy,: 
\begin{lem}\label{sotrois}
Every closed subgroup of $\sot$ is conjugate to precisely one group of the following list: 
\begin{equation*}
\{\id,\: \CC_n,\: \DD_n,\: \TT,\: \OO,\: \II,\: \sod,\: \ode,\: \sot\} 
\end{equation*}
\end{lem}
\noindent Among these groups, we can distinguish:
\begin{description}
\item[planar groups] : $\{\id,\: \CC_n,\: \DD_n,\: \sod,\: \ode\}$ which are $\ode$-closed subgroups;
\item [exceptional groups]: $\{\TT,\: \OO,\: \II,\: \sot\}$  which are the symmetry groups of platonician polyhedrons.
\end{description}
Let us detail first the set of planar subgroups. If we fix a base $(\ii;\jj;\kk)$ of $\GR^3$, and note $\QQ(\ve;\theta)\in \sot$ the rotation about $\ve\in \GR^3$ and of angle $\theta \in [0;2\pi[$ we have
\begin{listerd}
\item $\id$ the identity;
\item $\CC_n$ ($n\geq 2$) the cyclic group of order $n$, generated by the $n$-fold rotation $\QQ(\kk;\theta=\frac{2\pi}{n})$, is the symmetry group of a chiral polygon;
\item $\DD_n$ ($n\geq 2$) the dihedral group of order $2n$ generated by $\CC_n$ and $\QQ(\ii;\pi)$, is the symmetry group of a regular polygon;
\item $\sod$ the subgroup of rotations $\QQ(\kk;\theta)$ with $\theta \in [0;2\pi[$; 
\item $\ode$ the subgroup generated by $\sod$ and $\QQ(\ii;\pi)$.
\end{listerd}
Classes of exceptional subgroups are: $\TT$ the tetrahedral group of order $12$ which fixes a tetrahedron, $\OO$ the octahedral group of order $24$ which fixes an octahedron (or a cube), and $\II$ the subgroup of order $60$ which fixes an icosahedron (or a dodecahedron).

\noindent In \autoref{IsoConst}, the symmetry classes of $\mathbb{S}^{2n}$ and $\mathbb{G}^{2n}$ are obtained: 
\begin{lem}
The symmetry classes of $\mathbb{S}^{2n}$ are :  
\begin{eqnarray*}
\mathfrak{I}(\mathbb{S}^{2})&=&\lbrace [\DD_2],[\ode],[\sot]\rbrace\\
\mathfrak{I}(\mathbb{S}^{4})&=&\lbrace [\id],[\CC_2],[\DD_2],[\DD_3],[\DD_{4}],[\ode],[\OO],[\sot]\rbrace\\
n\geq3,\ \mathfrak{I}(\mathbb{S}^{2n})&=&\lbrace [\id],[\CC_2],\cdots,[\CC_{2(n-1)}],[\DD_2],\cdots,[\DD_{2n}],[\ode],[\TT],[\OO],[\II],[\sot]\rbrace
\end{eqnarray*}
\end{lem}
\begin{lem}
The symmetry classes of $\mathbb{G}^{2n}$ are : 
\begin{eqnarray*}
\mathfrak{I}(\mathbb{G}^{2})&=&\lbrace [\id],[\CC_2],[\DD_2],[\sod],[\ode],[\sot]\rbrace\\
\mathfrak{I}(\mathbb{G}^{4})&=&\lbrace [\id],[\CC_2],\cdots,[\CC_{4}],[\DD_2],\cdots,[\DD_{4}],[\sod],[\ode],[\TT],[\OO],[\sot]\rbrace\\
n\geq3,\ \mathfrak{I}(\mathbb{G}^{2n})&=&\lbrace [\id],[\CC_2],\cdots,[\CC_{2n}],[\DD_2],\cdots,[\DD_{2n}],[\sod],[\ode],[\TT],[\OO],[\II],[\sot]\rbrace
\end{eqnarray*}
\end{lem}
\noindent With the following cardinal properties:
\ben
\begin{tabular}{|c||c|c|c|}
  \hline
  $n$&1&2&$\geq$ 3\\ \hline
  $\#\mathfrak{I}(\mathbb{S}^{2n})$ & 3 & 8 & $2(2n+1)$  \\ \hline
  $\#\mathfrak{I}(\mathbb{G}^{2n})$ & 6& 12& $4n+5$\\ \hline
\end{tabular}
\een
The symmetry classes of $\mathbb{T}^{2n}$ are clarified by the following theorem :
\begin{thmI}
Let $\mathbb{T}^{2n}$ be a tensor space then either $\mathfrak{I}(\mathbb{T}^{2n})=\mathfrak{I}(\mathbb{S}^{2n})$ or
$\mathfrak{I}(\mathbb{T}^{2n})=\mathfrak{I}(\mathbb{G}^{2n})$.
\end{thmI}
\noindent In other terms,the number and the type of classes are the same as
\begin{itemize}
\item either $\mathbb{S}^{2n}$ the space of $2n$-order completely symmetric tensors. In this case, the number of classes is minimal;
\item or $\mathbb{G}^{2n}$ the space of $2n$-order generic tensors.  In this case the number of classes is maximal.
\end{itemize}
\noindent In fact, as specified by the following theorems, in most situations the number of class is indeed maximal: 

$\bullet$ For coupling tensors:
\begin{thmII}\label{th:SymCou}
Let consider $\mathbb{T}^{2p}$  the space of coupling tensors between two physics described respectively by two tensors vector spaces $\mathbb{E}_{1}$ and $\mathbb{E}_{2}$. If these tensor spaces are of orders greater or equal to $1$, then $\mathfrak{I}(\mathbb{T}^{2p})=\mathfrak{I}(\mathbb{G}^{2p})$.
\end{thmII}

$\bullet$ For proper tensors:
\begin{thmIII}\label{th:SymPro}
Let consider $\mathbb{T}^{2p}$,  the space of tensors of a proper physics described by the tensor vector space $\mathbb{E}$. If this tensor space is of order $p\geq3$, and is solely defined in terms of its index symmetries, then $\mathfrak{I}(\mathbb{T}^{2p})=\mathfrak{I}(\mathbb{G}^{2p})$.
\end{thmIII}

\begin{rem}
Exception occurs for:
\begin{description}
\item [$p=1$] the space of symmetric second order tensors is obtained;
\item [$p=2$] in the case of $\mathbb{T}_{(ij)}$, the space of elasticity tensors is obtained.
\end{description}
In each of the aforementioned situations the number of classes is minimal. There is no other situation where this case occurs. It should therefore be concluded that the space of elasticity tensors is exceptional.
\end{rem}

\subsection{Second strain-gradient elasticity (SSGE)}
\label{s:MSGE}
Application of the former theorems will be made on the even order tensors of SSGE. In first time the constitutive equations will be summed-up, then results will be stated. It worth noting, that obtaining the same results with the Forte-Vianello approach would have been far more complicated.

\subsubsection*{Constitutive laws}
In the second strain-gradient theory of linear elasticity \cite{Min65,FCB11}, the constitutive law gives the
symmetric Cauchy stress tensor\footnote{Exceptionally, in this subsection, tensor orders will be indicated by in-parenthesis exponents.} $\mathbf{\sigma}^{(2)}$\ and the hyperstress tensors $\mathbf{\tau}^{(3)}$ and $\mathbf{\omega}^{(4)}$ in terms of the infinitesimal strain tensor $\mathbf{\varepsilon}^{(2)}$\ and its gradients $\mathbf{\eta}^{(3)}=\mathbf{\varepsilon}^{(2)}\otimes\nabla$ and $\mathbf{\kappa}^{(4)}=\mathbf{\varepsilon}^{(2)}\otimes\nabla\otimes\nabla$ through the three linear relations:%
\begin{equation}
\begin{cases}
\mathbf{\sigma}^{(2)}=\mathbf{E}^{(4)}:\mathbf{\varepsilon}^{(2)}+\mathbf{M}^{(5)}\therefore\mathbf{\eta}^{(3)}+\mathbf{N}^{(6)}::\mathbf{\kappa}^{(4)}, \\ 
\mathbf{\tau}^{(3)}=\mathbf{M}^{T(5)}:\mathbf{\varepsilon}+\mathbf{A}^{(6)}\therefore\mathbf{\eta}^{(3)}+\mathbf{O}^{(7)}::\mathbf{\kappa}^{(4)},\\
\mathbf{\omega}^{(4)}=\mathbf{N}^{T(6)}:\mathbf{\varepsilon}^{(2)}+\mathbf{O}^{T(7)}\therefore\mathbf{\eta}^{(3)}+\mathbf{B}^{(8)}::\mathbf{\kappa}^{(4)}
\end{cases}
\label{SGE}
\end{equation}
where $:,\therefore,::$ denote, respectively, the double, third and fourth contracted product.
Above\footnote{The comma classically indicates the partial derivative with respect to spatial coordinates. The $T$ exponent denotes transposition. The transposition is defined by permuting the $p$ first indices with the $q$ last, where $p$ is the tensorial order of the image of a $q$-order tensor.}, $\sigma _{(ij)}$, $\varepsilon_{(ij)}$, $\tau _{(ij)k}$, $\eta_{(ij)k}=\varepsilon_{(ij),k}$, $\omega_{(ij)(kl)}$  and $\kappa_{(ij)(kl)}=\varepsilon_{(ij),(kl)}$ are, respectively, the matrix components of $\mathbf{\sigma}^{(2)}$, $\mathbf{\varepsilon}^{(2)}$, $\mathbf{\tau}^{(3)}$, $\mathbf{\eta}^{(3)}$, $\mathbf{\omega}^{(4)}$  and $\mathbf{\kappa}^{(4)}$ relative to an orthonormal basis $(\ii;\jj;\kk)$ of $\GR^3$.
And $E_{\underline{(ij)}\ \underline{(lm)}}$, $M_{(ij)(lm)n}$, $N_{(ij)(kl)(mn)}$, $A_{\underline{(ij)k}\ \underline{(lm)n}}$, $O_{(ij)k(lm)(no)}$ and $B_{\underline{(ij)(kl)}\ \underline{(mn)(op)}}$ are the matrix components
of the related elastic stiffness tensors.
\subsubsection*{Symmetry classes}

The symmetry classes of the elasticity tensors and of first strain-gradient elasticity tensors has been studied in \cite{FV96} and \cite{LAH+12}. Hence, here  solely considered the spaces of coupling tensors $\mathbf{N}^{(6)}$ and of second strain-gradient elasticity tensors $\mathbf{B}^{(8)}$ will be considered.\\

$\bullet$ Let define $\mathbb{C}\mathrm{es}$ to be the space of coupling tensors between classical elasticity and second strain-gradient elasticity:
\begin{equation*}
\mathbb{C}\mathrm{es}=\{\mathbf{N}^{(6)}\in\mathbb{G}^6|N_{(ij)(kl)(mn)}\}
\end{equation*}
A direct application of theorem \ref{th:SymCou} leads to the following results:
\begin{equation*}
\mathfrak{I}(\mathbb{C}\mathrm{es})=\lbrace [\id],[\CC_2],\cdots,[\CC_{6}],[\DD_2],\cdots,[\DD_{6}],[\sod],[\ode],[\TT],[\OO],[\II],[\sot]\rbrace
\end{equation*}
Therefore $\mathbb{C}\mathrm{es}$ is divided into $17$ symmetry classes.\\

$\bullet$ Let define $\mathbb{S}\mathrm{gr}$ to be the space of second strain-gradient elasticity tensors:
\begin{equation*}
\mathbb{S}\mathrm{gr}=\{\mathbf{O}^{(8)}\in\mathbb{G}^8|O_{\underline{(ij)(kl)}\ \underline{(mn)(op)}}\}
\end{equation*}
A direct application of theorem \ref{th:SymPro} leads to the following results:
\begin{equation*}
\mathfrak{I}(\mathbb{S}\mathrm{gr})=\lbrace [\id],[\CC_2],\cdots,[\CC_{8}],[\DD_2],\cdots,[\DD_{8}],[\sod],[\ode],[\TT],[\OO],[\II],[\sot]\rbrace
\end{equation*}
Therefore $\mathbb{S}\mathrm{gr}$ is divided into $21$ symmetry classes.

\section{Mathematical framework}\label{s:GenFra}

In this section the mathematical framework of symmetry analysis is introduced. In the first two subsections the notions of symmetry group and class are introduced, meanwhile the last is devoted to the introduction of irreducible spaces. The presentation is rather general, and will be specialized to tensor spaces only at the end of the section.

\subsection{Isotropy/symmetry groups}\label{ss:LinAct}

Let $\rho$ be a representation of a compact real Lie group\footnote{In the following $G$ will solely represent a \emph{compact} real Lie group, therefore this precision will mostly be omitted.} $G$ 
\begin{equation*}
    \rho: G \to \GL(\mathcal{E})
\end{equation*}
on a finite dimensional $\GR$-linear space $\mathcal{E}$. This action will be noted 
\[
\forall (g,\xx)\in G\times \mathcal{E},\quad g\cdot \xx=\rho(g)(\xx)
\]
For any element of $\mathcal{E}$, the set of operations $g$ in $G$ letting this element invariant is defined as
\[
\Sigma_{\xx}:=\lbrace g\in G \: \vert \: g\cdot \xx=\xx\rbrace
\]
This set, for physicists, is called the \emph{symmetry group} of $\xx$, and for mathematicians the \emph{stabilizer} or the \emph{isotropy subgroup} of $\xx$. It is worth noting that, owning to $G$-compactness, every isotropy subgroup is a closed subgroup of $G$. 
Conversely, a dual notion can be defined for $G$-elements. For any subgroup $K$ of $G$, the set of $K$-invariant elements in $\mathcal{E}$ is defined as
\[
\mathcal{E}^K:=\lbrace \xx\in \mathcal{E}\: \vert \: \ k\cdot \xx=\xx \;\  \forall k\in K\rbrace 
\]
Such a set is referred to as a \emph{fixed point set} and is a linear subspace of $\mathcal{E}$.  In this context we will note $\text{d}(K)=\text{dim } \mathcal{E}^K$.
It has to be observed that fixed-point sets  are group inclusion reversing, i.e. for subgroups $K_1$ and $K_2$ of $G$, we have the following property
\[
K_1\subset K_2 \Rightarrow \mathcal{E}^{K_2}\subset \mathcal{E}^{K_1}
\]
For a given isotropy group $K$, the former sets are linked by the following property:
\[
\xx \in \mathcal{E}^K \Rightarrow K\subset \Sigma_{\xx}
\] 
\subsection{Isotropy/symmetry classes}\label{ss:Red}

We aim at describing objects that have the same symmetry properties but may differ by their orientations in space. The first point is to define the set of all the positions an object can have. 
To that aim we consider the $G$-orbit of an element $\xx$ of $\mathcal{E}$: 
\[
\text{Orb}(\xx):=\lbrace g\cdot \xx \: \vert \: g\in G\rbrace \subset \mathcal{E}
\]
Due to $G$-compactness this set is a submanifold of $\mathcal{E}$. Elements of $\text{Orb}(\xx)$ will be said to be $G$-related. A fundamental observation is that $G$-related vectors have conjugate symmetry groups. More precisely\footnote{With the classical coset notation if $H$ is a subgroup of $G$ and $g\in G$ that is not in the subgroup $H$, then a left coset of $H$ in $G$ is defined
\ben
gH = \{gh : h \in H\} 
\een
and symmetricaly for a right coset.}
\beq\label{eq:SymCla}
\text{Orb}(\xx)=\text{Orb}(\yy) \Rightarrow \exists g\in G \: \vert \: \Sigma_{\xx}=g\Sigma_{\yy}g^{-1}
\eeq
Let us define the conjugacy class of a subgroup $K\subset G$ by
\beq
[K]=\{K'\subset G|\exists g\in G,K'= gK g^{-1}\}
\eeq
An isotropy class (or symmetry class) $[\Sigma]$ is defined as the conjugacy class of an isotropy subgroup $\Sigma$. This definition implies that there exists a vector $\xx\in \mathcal{E}$ such that $\Sigma=\Sigma_{\xx}$ and $\Sigma'\in [\Sigma]$; furthermore $\Sigma'=g\Sigma g^{-1}$ for some $g\in G$. 
The notion of isotropy class is the \emph{good} notion to define the symmetry property of an object modulo its orientation: symmetric group is related to a specific vector, but we deal with orbits, which are related to symmetric classes because of~\ref{eq:SymCla}.
Due to $G$-compactness there is only a finite number of isotropy classes \cite{Bre72}, and we note 
\[
\mathfrak{I}(\mathcal{E}):=\lbrace [\id];[\Sigma_1];\cdots;[\Sigma_l] \rbrace
\]
the set of all isotropy classes. In the case $G=\sot$  this result is known as the Hermann theorem \cite{Her45,Auf08}. Elements of $\mathfrak{I}(\mathcal{E})$ are conjugate to $\sot$-closed subgroups and this collection was introduced in \autoref{ss:OtrSub}.



\subsection{Irreducible spaces}

For every linear subspace $\mathcal{F}$ of $\mathcal{E}$, we note 
\[
g\cdot \mathcal{F}:=\lbrace g.\xx \: \vert \: \forall g\in G \:; \: \forall \xx\in \mathcal{F}\rbrace
\]
and we say that $\mathcal{F}$ is $G$-stable if $g\cdot \mathcal{F}\subset \mathcal{F}$ for every $g\in G$. It is clear that, for every representation, the subspaces $\{0\}$ and $\mathcal{E}$ are always $G$-stable. 
If, for a representation $\rho$ on $\mathcal{E}$, the only $G$-invariant spaces are the proper ones, the representation will be said irreducible. For a compact Lie group, the Peter-Weyl theorem \cite{Ste94} ensures that every representation can be split into a direct sum of irreducible ones. Furthermore, in the case $G=\sot$, those irreducible representations are explicitly known.\\

There is a natural action of $\sot$ on the space of $\GR^3$-\emph{harmonic polynomials}. If $p$ denotes such a harmonic polynomial, and $\xx \in \GR^3$, then, for every $g\in \sot$ we note
\[
g\cdot p (\xx)=p(g^{-1}\cdot \xx)
\]
Harmonic polynomials form a graded vector space, and to each subspace of a given degree a $\sot$-irreducible representation is associated.
$\mathcal{H}^k$ will be the vector space of harmonic polynomials of degree $k$, with $\dim\mathcal{H}^k=2k+1$.
If we take a vector space $V$ to be a $\sot$-representation, it can be decomposed into $\sot$-irreducible spaces:
\[
V=\bigoplus \mathcal{H}^{k_i}
\]
Grouping the same order irreducible spaces, the $\sot$-isotypic decomposition is obtained:
\ben\label{eq:decompgeneralP}
V=\bigoplus_{i=0}^{n}\alpha_{i}\mathcal{H}^{i}
\een
where  $\alpha_{i}$ is the multiplicity of the irreducible space $\mathcal{H}^{i}$ in the decomposition, \textbf{and $n$ the order of the higher order irreducible space of the decomposition.}\\

\subsubsection*{Application to tensor spaces}

In mechanics, $V$ is a vector subspace of $\otimes^{p}\RR^{3}$. In $\mathbb{R}^3$ there exists an isomorphism,$\phi$, between harmonic polynomial spaces and \textit{harmonic tensor spaces}\cite{Bac70,FV96}. Therefore all that have been previously said in term of harmonic polynomials can be translated in terms of harmonic tensors. A detailed discussion on this isomorphism can be found in \cite{Bac70}.
%
Therefore $\mathbb{H}^{k}=\varphi(\mathcal{H}^k)$, is the space of \textit{harmonic tensors}, that is the space of completely symmetric and traceless tensors. According to this isomorphism, any tensor space $\mathbb{T}^n$ can be decomposed into $\sot$-irreducible tensors:
\ben
\mathbb{T}^{n}=\bigoplus_{i=0}^{n}\alpha_{i}\mathbb{H}^{i}
\een
The symmetry group of $\mathbf{T}\in\mathbb{T}^n$ is the intersection of the symmetry groups of all its harmonic components\footnote{In the notation $\mathrm{H}^{i,j}$ the first exponent refers to the order of the harmonic tensor, meanwhile the second indexes the multiplicity of $\mathrm{H}^{i}$ in the decomposition.} 
\ben\label{eq:decompgeneralT}
\Sigma_{\mathbf{T}}=\bigcap_{i=0}^{n}\left(\bigcap_{j=0}^{\alpha_{i}}\Sigma_{\mathrm{H}^{i,j}}\right)
\een
In the same way, $\mathfrak{I}(\mathbb{T}^{n})$ will be obtained as a function of the symmetry classes of the irreducible representations involved in the harmonic decomposition of $\mathbb{T}^{n}$. The symmetry classes of $\sot$-irreducible representations are explicitly known \cite{GSI84,GSS88},  what is unknown is how to combine these results to determine the symmetry classes of $V$ (or $\mathbb{T}^{n}$).

\section{Clips operations}\label{ss:clips}
The aim of this section is to construct symmetry classes of a reducible representation from the irreducible ones. To that aim a new class-operator, named \emph{clips operator}, will be defined. The main result of this section is given in table \ref{global}, which contains all clips operations between $\sot$-closed subgroups. It is worth noting that this table contains more results than strictly needed for the proofs of our theorems. Nevertheless, we believe that these results are interesting in their own and may find applications for other problems. The explicit proofs of these results can be found in \autoref{calculs}.\\

In this section the intersection of solely two symmetry classes, is considered. Extensions to more general reducible representations will be treated in \autoref{lastone}.  
Let us start with the following lemma: 
\begin{lem}\label{sommeinter}
Let $\mathcal{E}$ be a representation of a compact Lie group $G$ that split into a direct sum of two $G$-stable subspaces 
\[
\mathcal{E}=\mathcal{E}_1\oplus \mathcal{E}_2 \text{ where } g\cdot \mathcal{E}_1\subset \mathcal{E}_1 \text{ and } g\cdot \mathcal{E}_2\subset \mathcal{E}_2 \: \forall g\in G 
\]
If we note by $\mathfrak{I}$ the set of all isotropy classes associated to $\mathcal{E}$, $\mathfrak{I}_i$ the set of all isotropy classes associated to $\mathcal{E}_i$ ($i=1,2$), then $[\Sigma] \in \mathfrak{I}$ if and only if there exist $[\Sigma_1]\in \mathfrak{I}_1$ and $[\Sigma_2]\in \mathfrak{I}_2$ such as $\Sigma =\Sigma_1 \cap \Sigma_2$.
\end{lem}

\begin{proof}
If we take $[\Sigma_1]\in \mathfrak{I}_1$ and $[\Sigma_2]\in \mathfrak{I}_2$, we know there exists two vectors $\xx_1\in \mathcal{E}_1$ and $\xx_2\in \mathcal{E}_2$ such that $\Sigma_i=\Sigma_{\xx_i}$ ($i=1,2)$. Then, let $\xx:=\xx_1+\xx_2$. 

For every $g\in  \Sigma_1 \cap \Sigma_2$ we have $g\cdot \xx_1+g\cdot \xx_2=\xx_1+\xx_2=\xx$; thus $\Sigma_1 \cap \Sigma_2
\subset \Sigma_{\xx}$. Conversely for every $g\in \Sigma_{\xx}$ we have 
\[
g\cdot \xx=\xx=g\cdot \xx_1+g\cdot \xx_2
\]
But, as $\mathcal{E}_i$ are $G$-stable and are in direct sum, we conclude that $g\cdot \xx_i=\xx_i$ ($i=1,2$). The reverse inclusion is proved. 

The other implication is similar: if we take $[\Sigma] \in \mathfrak{I}$ then we have $\Sigma=\Sigma_{\xx}$ for some $\xx\in \mathcal{E}$. And $\xx$  can be decomposed into $\xx_1+\xx_2$. The same proof as above shows that $\Sigma =\Sigma_{\xx_1} \cap \Sigma_{\xx_2}$. 
\end{proof}

Lemma \ref{sommeinter} shows that the isotropy classes of a direct sum are related to intersections of isotropy subgroups. But as intersection of classes is meaningless, the results cannot be directly extended. To solve this problem, one tool named \emph{clips operator} will be introduced. To that aim, let first  consider the following lemma: 

\begin{lem}
For every two $G$-classes $[\Sigma_i]$ ($i=1,2$), and for every $g_1$, $g_2$ in $G$, there exists $g=g_1^{-1}g_2$ in $G$ such that
\[
\left[ g_1\Sigma_1 g_1^{-1} \cap g_2\Sigma_2 g_2^{-1}\right]=[\Sigma_1 \cap g\Sigma_2 g^{-1}]
\]
\end{lem} 

\begin{proof}
Let $g=g_1^{-1}g_2$ and 
\[
\Sigma=g_1\Sigma_1 g_1^{-1} \cap g_2\Sigma_2 g_2^{-1}
\]
For every $\gamma \in \Sigma$ we have $\gamma=g_1 \gamma_1 g_1^{-1}=g_2\gamma_2 g_2^{-1}$ for some $\gamma_i\in \Sigma_i$ ($i=1,2$); then
\[
g _1\gamma g_1^{-1}=\gamma_1\in \Sigma_1 \text{ and } g_1\gamma g_1^{-1}=g\gamma_2 g^{-1}\in g\Sigma_2 g^{-1}
\]
Thus we have $g_1 \Sigma g_1^{-1}\subset \Sigma_1 \cap g\Sigma_2 g^{-1}$, and conversely. As $g_1 \Sigma g_1^{-1}$ is conjugate to $\Sigma$, we have proved the lemma. 
\end{proof}

\begin{defn}[Clips Operator]\label{clips}
For every $G$-classes $[\Sigma_1]$ and $[\Sigma_2]$, we define the \emph{clips operator} of $[\Sigma_1]$ and $[\Sigma_2]$, noted $[\Sigma_1] \circledcirc [\Sigma_2]$, to be 
\[
[\Sigma_1] \circledcirc [\Sigma_2]:=\lbrace [\Sigma_1\cap g\Sigma_2 g^{-1}]\: \: \text{ for all } g\in G \rbrace 
\]
which is a subset of $G$-classes. 
\end{defn}
\noindent If we note $\id$ the identity subgroup, one can observe some immediate properties: 

\begin{prop}\label{neutreclips}
For every $G$-class $[\Sigma]$ we have
\[
[\id]\circledcirc [\Sigma]=\lbrace [\id] \rbrace \text{ and } [G]\circledcirc [\Sigma]=\lbrace [\Sigma] \rbrace
\]
\end{prop}
\noindent Then, if we have two $G$-representations $\mathcal{E}_1$ and $\mathcal{E}_2$, and if we note $\mathfrak{I}_i$ the set of all isotropy classes of $\mathcal{E}_i$, the clips operator can be extended to these sets and noted 
\[
 \mathfrak{I}_1 \circledcirc \mathfrak{I}_2:=\bigcup_{\Sigma_1\in \mathfrak{I}_{1},\Sigma_2\in \mathfrak{I}_2} [\Sigma_1] \circledcirc [\Sigma_2]
\]
Then, by lemma \ref{sommeinter}, we obtain the corollary: 

\begin{cor}\label{correduction}
For every two $G$-representations $\mathcal{E}_1$ and $\mathcal{E}_2$, if $\mathfrak{I}_1$ denotes the isotropy classes of $\mathcal{E}_1$ and $\mathfrak{I}_2$ the isotropy classes of $\mathcal{E}_2$, then $\mathfrak{I}_1 \circledcirc \mathfrak{I}_2$ are all the isotropy classes of $\mathcal{E}_1\oplus \mathcal{E}_2$. 
\end{cor}

\begin{thm}\label{thm:clipssot}
For every two $\sot$-closed subgroups $\Sigma_1$ and $\Sigma_2$, we have $\id \in [\Sigma_1]\circledcirc [\Sigma_2]$. 
The other classes of the clips operation $[\Sigma_1]\circledcirc [\Sigma_2]$ are given in the 
table~\ref{global}
\end{thm}

%

\newpage

\begin{table}[h]
\caption{Clips operations on $\sot$-subgroups}\label{global}
\renewcommand{\arraystretch}{1.5}
\begin{center}

\begin{tabular}{|c|c|c|c|c|c|c|c|}
\hline
$\circledcirc$ & $\left[\CC_n\right]$ & $\left[\DD_n\right]$ & $\left[\TT\right]$ & $\left[\OO\right]$ & $\left[\II\right]$ & $\left[\sod\right]$ & $\left[\ode\right]$ \\
\hline
$\left[\CC_m\right]$ & $\left[\CC_d\right]$ & & & & & & \\
\cline{1-3}
$\left[\DD_m\right]$ & \begin{tabular}{l} $\left[ \CC_{d_2}\right]$ \\ $\left[ \CC_d\right]$ \end{tabular}& \begin{tabular}{l} $\left[\CC_{d_2}\right]$ \\ $\left[\CC_{d_2'}\right],\left[\CC_{dz}\right]$ \\ $\left[\CC_{d}\right],\left[\DD_{d}\right]$ \end{tabular} & & & & & \\
\cline{1-4}
$\left[\TT\right]$ & \begin{tabular}{l} $\left[\CC_{d_2}\right]$ \\ $\left[\CC_{d_3}\right]$ \end{tabular} & \begin{tabular}{l} $ \left[ \CC_{2}\right]$ \\ $\left[ \CC_{d_3}\right],\left[\DD_{d_2}\right]$ \end{tabular} & \begin{tabular}{l} $\left[ \CC_2\right]$ \\ $\left[ \CC_3\right]$\\ $\left[ \TT \right]$ \end{tabular} & & & & \\
\cline{1-5}
$\left[\OO\right]$ & \begin{tabular}{l} $\left[\CC_{d_2}\right]$ \\ $\left[\CC_{d_3}\right]$ \\ $\left[\CC_{d_4}\right]$ \end{tabular} & \begin{tabular}{l} $\left[\CC_{2}\right]$ \\ $\left[\CC_{d_3}\right],\left[\CC_{d_4}\right]$ \\    $\left[\DD_{d_2}\right],\left[\DD_{d_3}\right]$ \\ $\left[\DD_{d_4}\right]$ \end{tabular} & \begin{tabular}{l} $\left[ \CC_2\right]$ \\ $\left[ \CC_3\right]$ \\ $\left[ \TT \right]$ \end{tabular} & \begin{tabular}{l} $\left[ \CC_2\right]$ \\ $\left[ \DD_2\right],\left[ \CC_3\right]$ \\ $\left[ \DD_3\right],\left[ \CC_4\right]$ \\ $\left[ \DD_4 \right],\left[ \OO \right]$ \end{tabular} & & & \\
\cline{1-6}
$\left[\II\right]$ & \begin{tabular}{l} $\left[ \CC_{d_2}\right]$ \\ $\left[ \CC_{d_3}\right]$ \\ $\left[ \CC_{d_5}\right]$  \end{tabular} & \begin{tabular}{l} $\left[ \CC_{2}\right]$ \\ $\left[ \CC_{d_3}\right],\left[ \CC_{d_5}\right]$ \\ $\left[ \DD_{d_2}\right]$ \\ $\left[ \DD_{d_3}\right],\left[ \DD_{d_5}\right]$ \end{tabular} & \begin{tabular}{l} $\left[ \CC_2\right]$ \\ $\left[ \CC_3\right]$ \\ $\left[ \TT \right]$ \end{tabular} & \begin{tabular}{l} $\left[ \CC_{2}\right]$ \\ $ \left[ \CC_3\right],\left[ \DD_3\right]$ \\ $\left[ \TT\right]$\end{tabular} & \begin{tabular}{l} $\left[ \CC_{2}\right]$  \\ $\left[ \CC_3\right],\left[ \DD_3\right]$ \\ $\left[ \CC_5\right],\left[ \DD_5\right]$ \\ $\left[ \II\right]$ \end{tabular} & & \\
\cline{1-7}
$\left[\sod\right]$ & $\left[ \CC_n\right]$ & \begin{tabular}{l} $\left[ \CC_{2}\right]$ \\ $\left[ \CC_n\right]$ \end{tabular} & \begin{tabular}{l} $\left[\CC_{2}\right]$ \\ $\left[\CC_{3}\right]$ \end{tabular} &  \begin{tabular}{l} $\left[ \CC_{2}\right]$ \\ $\left[ \CC_{3}\right],\left[ \CC_{4}\right]$  \end{tabular} & \begin{tabular}{l} $\left[ \CC_{2}\right]$  \\ $\left[ \CC_3\right],\left[ \CC_5\right]$ \end{tabular} & $\left[\sod\right]$ & \\
\cline{1-8}
$\left[\ode\right]$ & \begin{tabular}{l} $\left[ \CC_{d_2}\right]$ \\ $\left[ \CC_n\right]$ \end{tabular} &\begin{tabular}{l} $\left[ \CC_{2}\right]$ \\ $\left[ \DD_n\right]$ \end{tabular} &\begin{tabular}{l} $\left[ \DD_{2}\right]$ \\ $\left[ \CC_3\right]$ \end{tabular} &\begin{tabular}{l} $\left[ \DD_{2}\right]$ \\ $\left[ \DD_3\right],\left[ \DD_4\right]$ \end{tabular} &\begin{tabular}{l} $\left[ \DD_{2}\right]$ \\ $\left[ \DD_3\right],\left[ \DD_5\right]$ \end{tabular} &\begin{tabular}{l} $\left[ \CC_{2}\right]$ \\ $\left[ \sod\right]$ \end{tabular} & \begin{tabular}{l} $\left[ \CC_{2}\right]$ \\ $\left[ \ode\right]$ \end{tabular}\\
\hline
\end{tabular}
\end{center}
\begin{center}
\begin{tabular}{llll}
\multicolumn{3}{c}{Notations} \\
$\CC_1:=\DD_1:=\id$ & $d_2:=gcd(n,2)$ & $d_3:=gcd(n,3)$ & $d_5:=gcd(n,5)$ \\
$d'_2:=gcd(m,2)$ & $dz:=2$ if $d=1$, $dz=1$ otherwise \\
$d_4:=\begin{cases} 4 \text{ if } 4\mid n \\ 1 \text{ otherwise } \end{cases}$ 
\end{tabular}
\end{center}

\end{table}

\section{Isotropy classes of harmonic tensors}\label{lastone}

In this section, the construction of the symmetry classes of a reducible representation from its irreducible components will be studied.
To that aim, in \autoref{ss:IsoClaIrr} the main results concerning the symmetry classes of irreducible representations are summed-up. In \autoref{ss:IsoClaRed}, and using the results of the previous section, basic properties of reducible representations are obtained. These results will be used in \autoref{IsoConst} to prove the theorems stated in \autoref{s:MaiRes}. From now on, all the results will be expressed in terms of tensor spaces.

\subsection{Isotropy classes of irreducibles}\label{ss:IsoClaIrr}
The following result was obtained by Golubitsky and al. \cite{GSI84,GSS88}: 
\begin{thm}\label{isotropyirreducible}
Let $\sot$ acts on $\mathbb{H}^k$. The following groups are symmetry classes of $\mathbb{H}^k$: 
\begin{itemize}
\item[(a)] $\id$ for $k\geq 3$; 
\item[(b)] $\CC_n$ ($n\geq 2)$ for $\displaystyle{\begin{cases}  n\leq k \text{ when } k \text{ is odd} \\ n\leq \dfrac{k}{2} \text{ when } k \text{ is even}\end{cases} }$
\item[(c)] $\DD_n$ ($n\geq 2$) for $n\leq k$; 
\item[(d)] $\TT$ for $k=3,6$, $7$ or $k\geq 9$; 
\item[(e)] $\OO$ for $k\neq 1,2,3,5,7,11$; 
\item[(f)] $\II$ for $k=6,10,12,15,18$ or $k\geq 20$ and $k\neq 23,29$; 
\item[(g)] $\sod$ for $k$ odd; 
\item[(h)] $\ode$ for $k$ even;
\item[(i)] $\sot$ for any $k$. 
\end{itemize}
\end{thm}
For forthcoming purposes, let us introduce the following notations. For each integer $k$, we note:
\[
\Gamma_{\TT(k)}:=\begin{cases} \TT \text{ if } \TT\in \mathfrak{I}^k \\ \emptyset \text{ otherwise} \end{cases}\:; \: \Gamma_{\OO(k)}:=\begin{cases} \OO \text{ if } \OO\in \mathfrak{I}^k \\ \emptyset \text{ otherwise} \end{cases} \text{ and } \Gamma_{\II(k)}:=\begin{cases} \II \text{ if } \II\in \mathfrak{I}^k \\ \emptyset \text{ otherwise} \end{cases}
\]
\[
\Sigma(k):=\begin{cases} \sod \text{ if } \sod \in \mathfrak{I}^k \\ \emptyset \text{ otherwise} \end{cases}\:; \: \Omega(k):=\begin{cases} \ode \text{ if } \ode \in \mathfrak{I}^k \\ \emptyset \text{ otherwise} \end{cases}
\]
where $\mathfrak{I}^k$ is the set of symmetry classes of $\mathbb{H}^k$.

\subsection{Isotropy classes of direct sum}\label{ss:IsoClaRed}
We have this obvious lemma, directly deduced from theorem~\ref{isotropyirreducible} : 
\begin{lem}\label{isotropexce}
We have $\ \Gamma_{\TT(k)}\neq \emptyset \Rightarrow \lbrace [\DD_2],[\DD_3] \rbrace 
\subset \mathfrak{I}^k) \quad;\quad 
\Gamma_{\OO(k)}\neq \emptyset \Rightarrow \lbrace [\DD_2],[\DD_3],[\DD_4] \rbrace 
\subset \mathfrak{I}^k \quad;\quad 
\Gamma_{\II(k)}\neq \emptyset \Rightarrow \lbrace [\DD_2],\cdots,[\DD_5] \rbrace 
\subset \mathfrak{I}^k.
$
\end{lem}

We note $\mathfrak{I}(k,n)$, the $(n-1)$ self clips operations of $\mathfrak{I}^k$,which is the set of isotropy classes of a $n$-uple of $k$-th harmonic tensors\cite{AKP12}, i.e. $n \mathbb{H}^k$.  The basic operations are, for every integers $k\geq 1$ and $n\geq 2$ 
\[
\mathfrak{I}(k,n):=\mathfrak{I}^k\circledcirc \mathfrak{I}(k,n-1) \text{ and } \mathfrak{I}(k,1):=\mathfrak{I}^k
\]
On the simple example of  $\mathbb{H}^2$, the following fact can be observed
\[
\mathfrak{I}(2,n):=\mathfrak{I}^2\circledcirc \mathfrak{I}(2,n-1)=\mathfrak{I}^2\circledcirc \mathfrak{I}^2=\left\lbrace  \left[ \id \right],\left[ \CC_2 \right],\left[ \DD_2 \right],\left[ \ode \right],\left[ \sot \right]\right\rbrace
\]
This result can be generalized: 
\begin{cor}\label{multiirred}
For every integers $n\geq 2$ and $k\geq 2$, the isotropy classes of $n \mathbb{H}^k$ are
\begin{equation*}
\mathfrak{I}(k,n)=\mathfrak{I}^k\circledcirc \mathfrak{I}^k=\lbrace [\id],[\CC_2],\cdots,[\CC_k],[\DD_2],\cdots,[\DD_k],[\Gamma_{\TT(k)}],[\Gamma_{\OO(k)}],[\Gamma_{\II(k)}],[\Sigma(k)],[\Omega(k)],[\sot]\rbrace
\end{equation*}
\end{cor}

\begin{proof}
From theorem~\ref{isotropyirreducible} we know that $[\DD_l]\in \mathfrak{I}^k$ for $2\leq l \leq k$; furthermore $[\sot]\in \mathfrak{I}^k$ then from lemma~\ref{neutreclips} we know that, for all integers $2\leq l \leq k$ we will have (by induction), for all $n\geq 2$ $[\DD_l]\in \mathfrak{I}(k,n)$.
Then, when we compute $\mathfrak{I}(k,n)\circledcirc \mathfrak{I}^k$ we will have $[\DD_l]\circledcirc [\DD_l]=\lbrace [\id],[\CC_l],[\DD_l]\rbrace$
Neither $[\ode]$ and $[\sod]$, with cyclic or dihedral conjugacy classes, generate other cases. The same occurs for clips operation of cyclic groups. Now, because  of the lemma~\ref{isotropexce} we also see that no exceptional conjugacy class generates other cases. We have then prove our corollary. 
\end{proof}


\begin{cor}\label{deuxisospairs}
For every integers $2\leq 2p<2q$, we have 
\begin{eqnarray*}
\mathfrak{I}(2p,2q):=\mathfrak{I}^{2p}\circledcirc \mathfrak{I}^{2q}&=&\lbrace [\id],[\CC_2],\cdots,[\CC_{\text{max}(q;2p)}],[\DD_2],\cdots,[\DD_{2q}], \\
& &[\Gamma_T(2p)\cup \Gamma_T(2q)],[\Gamma_O(2p)\cup \Gamma_O(2q)],[\Gamma_I(2p)\cup \Gamma_I(2q)], \\
& &[\ode],[\sot]\rbrace
\end{eqnarray*}
\end{cor}

\begin{proof}
First of all, it is clear that because $[\sot]\in \mathfrak{I}^{k_i}$ ($i=1,2$) we will have all $[\DD_l]$ for $2\leq l \leq 2q$. Now we will have all $[\CC_i]\circledcirc [\sot]$ for all $1\leq i \leq p$. We also have $[\CC_j]\in [\DD_j]\circledcirc [\DD_j]$ for all $1\leq j \leq 2q$; this show that 
\[
\lbrace [\id],[\CC_2],\cdots,[\CC_{\text{max}(q;2p)}] \rbrace \subset \mathfrak{I}[2p,2q]
\]
Now, we can observe that clips operation of dihedral groups and $[\ode]$ does not generate cyclic groups; and lemma~\ref{isotropexce} shows that no other cases can be generated with exceptional subgroups. 
\end{proof}

%
\section{Isotropy classes of constitutive tensors}\label{IsoConst}

\subsection{The symmetry classes of even-order tensor space}

Let consider the constitutive tensor space $\mathbb{T}^{2n}$, it is known that this space can be decomposed orthogonally into a full symmetric space and a complementary one which is isomorphic to a tensor space of order $2n-1$ \cite{JCB78}, i.e. :
\begin{equation*}
\mathbb{T}^{2n}=\mathbb{S}^{2n}\oplus \mathbb{C}^{2n-1}
\end{equation*}
Let consider the $\sot$-isotypic decomposition of $\mathbb{T}^{2n}$
\begin{equation*}
\mathbb{T}^{2n}=\bigoplus_{k=0}^{2n}\alpha_{k}\mathbb{H}^{k},\quad\text{with}\ \alpha_{2n}=1
\end{equation*}
The part related to $\mathbb{S}^{2n}$ solely contains even order harmonic tensors with multiplicity one \cite{JCB78}, i.e.
\begin{equation*}
\mathbb{S}^{2n}=\bigoplus_{k=0}^{n}\mathbb{H}^{2k}\quad \mathrm{and}\quad \mathbb{C}^{2n-1}=\bigoplus_{k=0}^{2n-1}\alpha'_{k}\mathbb{H}^{k}\quad \text{with}\ 
\alpha'_{k}=
\begin{cases}
\alpha_{k}\ \text{for $k$ odd}\\
\alpha_{k}-1\ \text{for $k$ even}\\
\end{cases}
\end{equation*}
Using the clips operator, the symmetry classes of $\mathbb{T}^{2n}$ can be expressed:
\ben
\mathfrak{I}(\mathbb{T}^{2n}):=\mathfrak{I}(\mathbb{S}^{2n})\circledcirc\mathfrak{I}(\mathbb{C}^{2n-1})
\een
Let us first determined the symmetry classes of $\mathbb{S}^{2n}$. Using the results of the previous section, we have:
\begin{lem}\label{symTens}
\ban
n=1,\  \mathfrak{I}(\mathbb{S}^{2})&=&\lbrace [\DD_2],[\ode],[\sot]\rbrace\\
n=2,\  \mathfrak{I}(\mathbb{S}^{4})&=&\lbrace [\id],[\CC_2],[\DD_2],[\DD_3],[\DD_{4}],[\ode],[\OO],[\sot]\rbrace\\
\forall\ n\geq 3,\ 
\mathfrak{I}(\mathbb{S}^{2n})&=&\lbrace [\id],[\CC_2],\cdots,[\CC_{2(n-1)}],[\DD_2],\cdots,[\DD_{2n}],[\ode],[\TT],[\OO],[\II],[\sot]\rbrace
\ean
With the following cardinal properties: 
\ben
\#\mathfrak{I}(\mathbb{S}^{2})=3\quad;\quad \#\mathfrak{I}(\mathbb{S}^{4})=8\quad;\quad \#\mathfrak{I}(\mathbb{S}^{2n})=2(2n+1)
\een
\end{lem}

\begin{proof}
The case $n=1$ is obtained as a direct application of the theorem \ref{isotropyirreducible} and the proposition \ref{neutreclips}.
For $n\geq2$, let consider the corollary \ref{deuxisospairs} in the case $k_{1}=2(n-1)$ and $k_{2}=2n$. 
\begin{eqnarray*}
\mathfrak{I}(2(n-1),2n)&:=&\lbrace [\id],[\CC_2],\cdots,[\CC_{2(n-1)}],[\DD_2],\cdots,[\DD_{2n}],[\Gamma_T(2(n-1))\cup \Gamma_T(2n)], \\
 &&[\Gamma_O(2(n-1))\cup \Gamma_O(2n)],[\Gamma_I(2(n-1))\cup \Gamma_I(2n)],[\ode],[\sot]\rbrace 
\end{eqnarray*}
In the collection of planar isotropy classes, $[\CC_{2n-1}]$ and $[\CC_{2n}]$ are missing. It should be observed that the clips operation $\mathfrak{I}[2(n-1),2n]\circledcirc\mathfrak{I}^{2(n-2)}$ can never complete the sequence.
%
%

For exceptional groups it can be observed that for any $n\geq3$ the $\sot$-irreducible decomposition will contain $\mathbb{H}^{6}$. As $\{[\TT],[\OO],[\II]\}$ are isotropy classes for $\mathbb{H}^{6}$, it would be the same for any space that contains $\mathbb{H}^{6}$. 

$\bullet$ Therefore, for $n\geq3$ 
\begin{equation*}
\mathfrak{I}(\mathbb{S}^{2n})=\lbrace [\id],[\CC_2],\cdots,[\CC_{2(n-1)}],[\DD_2],\cdots,[\DD_{2n}],[\ode],[\TT],[\OO],[\II],[\sot]\rbrace
\end{equation*}
and, $\#\mathfrak{I}(\mathbb{T}_{s}^{2n})=2(2n-1)$.

$\bullet$ For the case $n=2$, we obtain the same result but without the class $\TT$ and $\II$ and, in such a case, $\#\mathfrak{I}(\mathbb{T}_{s}^{4})=8$.
\end{proof}


\begin{defn}
For a given $\sot$ representation on the tensor space $\mathbb{T}^{2n}$ ($n\geq 3$), we define 
\begin{eqnarray*}
\mathcal{C}(2n)&=&\lbrace [\id],[\CC_2],\cdots,[\CC_{2n}],[\DD_2],\cdots,[\DD_{2n}],[\sod],[\ode],[\TT],[\OO],[\II],[\sot]\rbrace
\end{eqnarray*}
We also define
\begin{eqnarray*}
n=1,\ \mathcal{C}(2)&=&\lbrace [\id],[\CC_2],[\DD_2],[\sod],[\ode],[\sot]\rbrace\\
n=2,\ \mathcal{C}(4)&=&\lbrace [\id],[\CC_2],\cdots,[\CC_{4}],[\DD_2],\cdots,[\DD_{4}],[\sod],[\ode],[\TT],[\OO],[\sot]\rbrace\\
\end{eqnarray*} 
\end{defn}

One can observe that these sets are in fact all the isotropy classes allowed by the Hermann theorem, and we clearly have the following cardinal properties: 
\ben
\# \mathcal{C}(2)=6\quad;\quad \# \mathcal{C}(4)=12\quad;\quad \# \mathcal{C}(2n,n\geq3)=4n+5
\een


\begin{defn}
Let $\mathbb{T}^{2n}$ be a tensor space which $\sot$-irreducible decomposition is  $\mathbb{T}^{2n}\simeq\bigoplus_{k=0}^{2n}\alpha_{k}\mathbb{H}^{k}$. $\mathbb{T}^{2n}$ is said to be \textit{even-harmonic} (EH) if $\alpha_{2p+1}=0$ for each $0\leq p\leq(n-1)$.  \\
\end{defn}

\begin{lem}
Let $\mathbb{G}^{2n}$ be the vector space of $2n$-th order tensors with no index symmetries, $\mathbb{G}^{2n}$ is not EH.
\end{lem}
\begin{proof}
For $n\geq1$, the induced reducible $\sot$-representation on $\mathbb{G}^{2n}=\otimes^{2n}\mathbb{R}^3$ is construct by tensorial products of the vectorial one. Such a construction implies odd-order tensors in the harmonic decomposition of $\mathbb{G}^{2n}$. 
\end{proof}

\noindent Now, the following proposition can be \textbf{proved}
\begin{thmI}
Let $\mathbb{T}^{2n}$ be a tensor space, for any $n\geq 3$, if $\mathbb{T}^{2n}$ is EH then $\mathfrak{I}(\mathbb{T}^{2n})=\mathfrak{I}(\mathbb{S}^{2n})$, otherwise $\mathfrak{I}(\mathbb{T}^{2n})=\mathfrak{I}(\mathbb{G}^{2n})$.
\end{thmI}

\begin{proof}
Let consider the $\sot$-irreducible decomposition of $\mathbb{T}^{2n}$, this decomposition can be written $\mathbb{T}^{2n}\simeq\mathbb{S}^{2n}\oplus\mathbb{C}^{2n-1}$.
The following inclusion is always true 
\begin{equation*}
\mathfrak{I}(\mathbb{S}^{2n})\subseteq\mathfrak{I}(\mathbb{T}^{2n})\subseteq\mathfrak{I}(\mathbb{G}^{2n})\subseteq\mathcal{C}(2n)
\end{equation*}
If $\mathbb{T}^{2n}$ is not EH, there exists \textbf{at least one} $k\in\mathbb{N}$ such as $\alpha_{2k+1}\neq0$
\begin{equation*}
\mathfrak{I}(\mathbb{S}^{2n})\circledcirc[\sod]\subseteq\mathfrak{I}(\mathbb{S}^{2n})\circledcirc\mathfrak{I}(\mathbb{C}^{2n-1})=\mathfrak{I}(\mathbb{T}^{2n})
\end{equation*}
since, from theorem \ref{isotropyirreducible}, any odd-order harmonic tensor contains $[\sod]$ in its symmetry classes.
As, from lemma \ref{symTens}, dihedral groups are contained up to $2n$ in $\mathfrak{I}(\mathbb{S}^{2n})$, the missing cyclic groups of $\mathfrak{I}(\mathbb{S}^{2n})$ are obtained by clips products with $[\sod]$.
Therefore 
\begin{equation*}
\mathfrak{I}(\mathbb{S}^{2n})\circledcirc\sod=\mathcal{C}(2n) \quad\text{hence}\quad \mathfrak{I}(\mathbb{T}^{2n})=\mathcal{C}(2n)
\end{equation*}
As $\mathbb{G}^{2n}$ is not EH, $\mathfrak{I}(\mathbb{G}^{2n})=\mathcal{C}(2n)$. Therefore we reach the sought conclusion, that if $\mathbb{T}^{2n}$ is not EH
$\mathfrak{I}(\mathbb{T}^{2n})=\mathfrak{I}(\mathbb{G}^{2n})$.

Conversely, if $\mathbb{T}^{2n}$ is  EH, $\mathbb{C}^{2n-1}$ solely contains even order irreducible spaces and its leading harmonic spaces are, at most, of order $2(n-1)$.
If the leading harmonic spaces order are strictly less than $2(n-1)$, the same analysis as for lemma.\ref{symTens} lead to the same conclusion.
Suppose now, that $\alpha'_{2(n-1)}\geq1$, using all the previous results,we have
\ban
\mathfrak{I}(\mathbb{T}^{2n})&=&\mathfrak{I}(\mathbb{S}^{2n})\circledcirc\mathfrak{I}(\mathbb{C}^{2n-1})\\
	&=&(\mathfrak{I}^{2n}\circledcirc\mathfrak{I}^{2(n-1)})\circledcirc(\mathfrak{I}(2(n-1),\alpha'_{2(n-1)})\circledcirc\mathfrak{I}(2(n-2),\alpha'_{2(n-2)}))\\
	&=&(\mathfrak{I}^{2n}\circledcirc\mathfrak{I}(2(n-1),2))=\mathfrak{I}(\mathbb{S}^{2n})\circledcirc\mathfrak{I}^{2(n-1)}
\ean
As $\mathfrak{I}^{2(n-1)}$ does not contain $[\sod]$ the missing classes cannot be generated, therefore $\mathfrak{I}(\mathbb{T}^{2n})=\mathfrak{I}(\mathbb{S}^{2n})$.
\end{proof}
Therefore it should be concluded that for any $2n$-order tensor space, symmetry classes are either the same as $\mathbb{S}^{2n}$ or the same as $\mathbb{G}^{2n}$.
Let us now investigate, in the next subsection, under which conditions a constitutive tensor space falls into one of the identified two possibilities.

\subsection{Construction of a constitutive tensor space}

\textbf{This last subsection will be devoted to the proof of our main result.} The space of constitutive tensors is a subspace of linear applications from $\mathbb{E}_{1}$ to $\mathbb{E}_{2}$. As seen in \autoref{s:MaiRes} 
\begin{equation*}
\mathcal{L}(\mathbb{E}_{1},\mathbb{E}_{2})\simeq \mathbb{E}_{1}\otimes \mathbb{E}_{2}\subset \mathbb{T}^{p}\otimes\mathbb{T}^{q}\simeq\mathbb{T}^{2n=p+q}
\end{equation*}
These vector spaces described the physical quantities involved in the problem understudy. We know, from the previous section, that any CTS has as many symmetry classes as either the complete symmetric tensor spaces, or the generic tensor space. 
Here we are interested in obtaining the conditions both on $\mathbb{E}_{1}$ to $\mathbb{E}_{2}$ and on the tensor product (symmetric or not) under which $\mathbb{T}^{2n}$ is even order, and therefore has a minimal number of symmetry classes. To that aim, distinction will be made between coupling and proper tensor spaces, in the sense previously defined in \autoref{s:MaiRes}.

\subsubsection*{Coupling tensor spaces}
We consider here two STS given by their $\sot$ isotypic decomposition:
\[
\mathbb{E}_{1}=\mathbb{T}^p=\bigoplus_{i=0}^p \beta_i\mathbb{H}^i \text{ and } \mathbb{E}_{2}=\mathbb{T}^q=\bigoplus_{i=0}^q \gamma_j\mathbb{H}^j \text{ with } \beta_p=\gamma_q=1
\]
Now results are given by the two following lemmas:
\begin{lem}
If $\mathbb{E}_{1}\neq \mathbb{E}_{2}$, $p>q$ and if $\mathbb{T}^p\otimes \mathbb{T}^q$ is EH, then $\mathbb{T}^{p}$ is EH and $\mathbb{T}^{q}=\mathbb{H}^{0}$;
\end{lem}
\begin{proof}
It is sufficient to consider the tensor product of the two leading irreducible spaces, and to use the Clebsch-Gordan product for $\sot$ \cite{JCB78,Auf08}. We obtain
\ben
\mathbb{H}^{p}\otimes\mathbb{H}^{q}=\bigoplus_{i=p-q}^{p+q}\mathbb{H}^{i}
\een
Therefore $p$ must be even, and $q=0$, therefore $\mathbb{T}^{q}=\gamma_{0}\mathbb{H}^{0}$, and by hypothesis $\gamma_{0}=1$. Thus $\mathbb{T}^{p}$ has to be EH.
\end{proof}
\begin{lem}
If $\mathbb{E}_{1}=\mathbb{E}_{2}$ (and then $p=q$) and if $\mathbb{T}^p\otimes \mathbb{T}^p$ is EH then $\beta_{i}=\gamma_{i}$. Furthermore, if $\mathcal{L}(E_1)$ is not self-adjoint then $\mathbb{T}^{p}=\mathbb{H}^{0}$;
\end{lem}
\begin{proof}
The demonstration is the same as the preceding proof.
\end{proof}
These results can be summed-up in the following theorem:
\begin{thmII}
Let consider $\mathbb{T}^{2p}$  the space of coupling tensors between two physics described respectively by two tensor vector spaces $\mathbb{E}_{1}$ and $\mathbb{E}_{2}$. If these tensor spaces are of orders greater or equal to $1$, then $\mathfrak{I}(\mathbb{T}^{2n})=\mathfrak{I}(\mathbb{G}^{2n})$
\end{thmII}
\begin{proof}
This is a direct application of the two former lemmas.
\end{proof}

\subsubsection*{Proper tensor spaces}
In this case we have the following lemma:
\begin{lem}\label{lem:CriSym}
If $\mathbb{E}=\mathbb{E}_{1}=\mathbb{E}_{2}$, $\mathbb{T}^p\otimes^{S}\mathbb{T}^p$ is EH and if $\mathcal{L}(\mathbb{E})$ is self-adjoint then
\begin{itemize}
\item if $p=2m+1$ then $\mathbb{T}^{2m+1}=\mathbb{H}^{2m+1}$;
\item if $p=2m$ then $\mathbb{T}^{2m}=\mathbb{H}^{2m}\oplus\beta_{0}\mathbb{H}^{0}$;
\end{itemize}
\end{lem}

\begin{proof}
Because $\mathcal{L}(\mathbb{E})$ is self-adjoint the tensor product is replaced by the symmetric tensor product, and if  $\mathbb{T}^p=
\bigoplus_{i=0}^{p}\beta_{i}\mathbb{H}^{i}$, the symmetric tensor product $\mathbb{T}^p\otimes^S \mathbb{T}^p$ can be decomposed into a direct sum of 
\[
\beta_i^2 \mathbb{H}^i\otimes^S \mathbb{H}^i \text{ and } \beta_i\beta_j\mathbb{H}^i\otimes \mathbb{H}^j \text{ with } i< j \in \lbrace 0,\cdots,p\rbrace
\]
with the following Clebsch-Gordan rule for the symmetric product :
\ben
\mathbb{H}^{k}\otimes^{S}\mathbb{H}^{k}=\bigoplus_{i=0}^{k}\mathbb{H}^{2i}
\een
Therefore, we cannot have the tensor product $\mathbb{H}^i\otimes \mathbb{H}^j$ for $1\leq i\leq p-1$ and $i\neq j$; thus we deduce that 
\[
\mathbb{T}^p=\beta_0\mathbb{H}^0\oplus \mathbb{H}^p \text{ and } \mathbb{T}^p\otimes^S \mathbb{T}^p=\beta_0 \mathbb{H}^0\oplus (\mathbb{H}^p\otimes^S \mathbb{H}^p)\oplus 2\beta_0 \mathbb{H}^p
\] 
Then : either $p$ is odd, and in such case $\beta_0=0$; either $p$ is even and then $\mathbb{T}^p=\beta_0\mathbb{H}^0\oplus \mathbb{H}^p$. 
\end{proof}

We therefore obtain the following result
\begin{thmIII}
Let consider $\mathbb{T}^{2p}$,  the space of tensors of a proper physics described by the tensor vector space $\mathbb{E}$. If this tensor space is of order $p\geq3$, and is solely defined in terms of its index symmetries, then $\mathfrak{I}(\mathbb{T}^{2n})=\mathfrak{I}(\mathbb{G}^{2n})$.
\end{thmIII}
\begin{proof}
Any tensor subspace defined in terms of its index symmetries contains, as a subspace, the space of full symmetric tensors. Since $p=3$, the harmonic decomposition of $\mathbb{S}^3$ does not satisfied the condition of lemma \ref{lem:CriSym}. Direct application of this lemma leads to the conclusion.
\end{proof}

\begin{rem}
It can be observed, that CTS having a minimum number of classes can nevertheless be constructed. They consists in spaces of autoadjoint linear applications between harmonic spaces, which are defined both from complete symmetry under index permutations and traceless property.
\end{rem}


\section{Conclusion}
\label{sec:conclusion}

In this paper the symmetry classes determination of even order tensors has been studied. Based on a new geometric approach, a complete and general answer to this recurrent problem in continuum mechanics has been given. Application of our results are direct and solve directly problems that would have been difficult to manage with the Forte-Vianello method.
As an example, and for the first time, the symmetry classes of the even-order tensors involved in Mindlin second strain-gradient elasticity  were given. 
To reach this goal a geometric tools, named the clips operator, has been introduced. Its main properties and all the products for $\sot$-closed-subgroups were also provided. We believe that these results may find applications in others context.
Using the geometrical framework introduced in this paper, some extensions of the current method can be considered:
\begin{itemize}
  \item To extend this approach to odd-order tensors;
  \item To take into account the coexistence of different symmetry properties for the physical properties of architectured  multimaterials;
\end{itemize}
These extensions will be the objects of forthcoming papers.

\appendix

\section{Clips operation on $\sot$-subgroups}\label{calculs}
In this section results concerning the clips operations on $\sot$-subgroups will be established.
The geometric idea to study the intersection of symmetry classes relies on the symmetry determination of composite figures which symmetry groups is the intersection of two elementary figures. 
As an example let consider the rotation $\rr=\QQ\left(\kk;\cfrac{\pi}{3}\right)$ , determining  $\DD_4 \cap \rr\DD_4 \rr^t$
is amount to establish the set of transformations letting the composite figure~\ref{cubimbrique} invariant.
\begin{figure}
	\centering
		\includegraphics{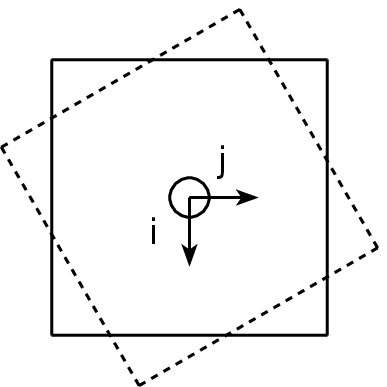}
\caption{Composite figure associated to $\DD_4 \cap \rr\DD_4 \rr^t$ where $\rr=\QQ\left(\kk;\cfrac{\pi}{3}\right)$}\label{cubimbrique}
\end{figure}


\subsection{Parameterization of subgroups}

We will define geometric elements for each $\sot$ closed subgroup: 
\begin{listerd}
\item The cyclic group $\CC_n$ is characterized by the $Oz$ axis; it will be noted $\CC_n^0:=\CC_n$;
\item The same convention is retained for the dihedral group $\DD_n$, i.e. $\DD_n^0:=\DD_n$;
\item For the cube $\mathcal{C}_0$ (c.f fig.\ref{cube0}) we defined its vertex collection $\{A_i\}_{i=1\cdots 8}={(\pm1;\pm1\pm1)}$, $\mathcal{C}_0$ is $\OO^0$-invariant; 
\item For the tetrahedron we consider the figure~\ref{cube0} and define $\mathcal{T}_0$ to be the tetrahedron $A_1A_3A_7A_5$, $\mathcal{T}_0$ is $\TT^0$-invariant. 
\item For the dodecahedron (c.f. fig.\ref{dode2}), we note $\mathcal{D}_0$ the figure which vertices are\footnote{With $\phi$ the golden ratio.}

\begin{figure}
    \begin{minipage}[t]{6cm}
        \centering
        \includegraphics[scale=0.8]{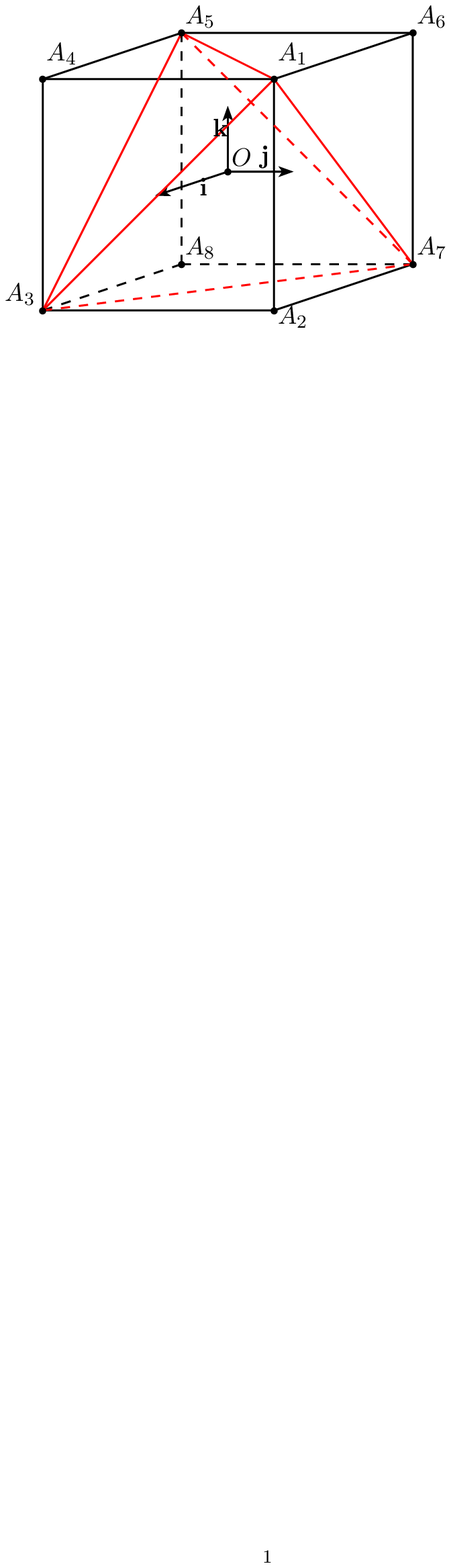}
        \caption{Cube $\mathcal{C}_0$}
        \label{cube0}
    \end{minipage}
    \begin{minipage}[t]{7cm}
        \centering
        \includegraphics[scale=0.8]{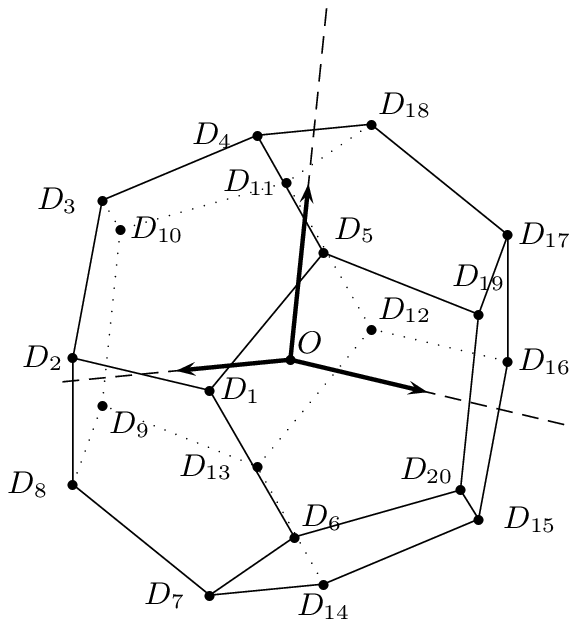}
        \caption{Dodecahedron $\mathcal{D}_0$}
        \label{dode2}
    \end{minipage}
\end{figure}


\begin{itemize}
\item Twelve vertices of type $\left( \pm \frac{a}{2},\pm \phi^2 \frac{a}{2},0\right)$ circularly permuted; 
\item Eight vertices of a cube with coordinates $\left( \pm \phi\frac{a}{2},\pm \phi\frac{a}{2},\pm \phi\frac{a}{2}\right)$
\end{itemize}
\end{listerd} 


\subsection{Axes and subgroup classes}\label{axessub}

For every $\sot$-subgroups, we defined its $g$ conjugate in the following way $K^g=gK^0 g^{t}$ where the $g$ exponent indicates the transformation, and $0$ the initial configuration. To proceed toward our analysis we need to introduce the following group decomposition \cite{GSI84,GSS88}
\begin{defn}\label{disjoint}
Let $K_1$, $K_2$, $\cdots$, $K_s$ be subgroups of $\Sigma$. Then $\Sigma$ is the \emph{direct union} of the $K_i$'s if 
\begin{equation*}
\text{a)}\ \Sigma=\bigcup_{i=1}^s K_i \quad ;\quad  \text{b)}\ K_i\cap K_j=\lbrace e \rbrace\  \forall i\neq j. 
\end{equation*}
\end{defn}
This decomposition is noted $K=\biguplus_{i=1}^s K_i$ when we have a \emph{direct union} of subgroups. 
We give some important details about geometric structure of $\sot$-subgroup. Indeed:
\begin{listerd}
\item  $\CC_n^0$ is characterized by the $Oz$ axis, generated by $\kk$. For every rotation $g\in \sot$, we note $a$ the axis generated by $g\kk$ and note $\CC_n^a=\CC_n^g$
to indicate the rotation axis.  
\item  $\DD_n^0$ is characterized by its primary axis $Oz$ and several secondary axis $b_l$. Therefore
\begin{equation}\label{diedre}
\DD_n^0=\CC_n^0 \biguplus_{l=0}^{n-1} \CC_2^{b_l}
\end{equation}
Each  $b_l$ is perpendicular to $Oz$ and are related by the $\CC_n^0$ generator.  $\DD_n^0$ is chosen such as one $b_l$ is generated by $\ii$. For every rotation $g\in \sot$ we define $a$ - generated by $g\kk$ - to be the primary axis and $b$ - generated by $g\ii$ - the secondary one, this is noted
\[
\DD_n^{a,b}=\DD_n^g
\]
\item The subgroup $\TT^0$ can be split into a direct union of cyclic subgroups \cite{GSI84}. 
\begin{equation}\label{dtetra}
\TT^0= \biguplus_{i=1}^4 \CC_3^{vt_i} \biguplus_{j=1}^3 \CC_2^{et_j}
\end{equation}
where the vertex axes of the tetrahedron are noted $vt_i$ and his edge axes are noted $et_j$, the details of these axis appear on figure~\ref{cube0}.
Each conjugate subgroup $\TT^g$ will be characterized by the set of its axes $(gvt_i,get_i)$, $g\in \sot$.
\item The octahedral subgroup $\OO^0$  split into 
\begin{equation}\label{deccube}
\OO^0=\biguplus_{i=1}^3 \CC_4^{fc_i} \biguplus_{j=1}^4 \CC_3^{vc_j} \biguplus_{l=1}^6 \CC_2^{ec_l}
\end{equation}
where vertex, edge and face axes are noted, respectively: $vc_i$, $ec_j$ and $fc_j$. Detail can be found on figure~\ref{cubeaxe}.
\begin{figure}[h]
\begin{center}
\includegraphics[scale=0.8]{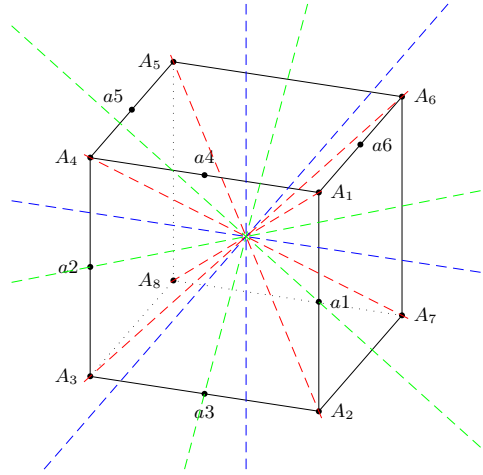}
\end{center}
\caption{Symmetry axis of $\mathcal{C}_0$}\label{cubeaxe}
\end{figure}
For every rotation $g\in \sot$, $\OO^g$ is characterized by its set of transformed axes ($gfc_i,gec_j,gvc_l$).
\item The icosahedral group $\II^0$ split into 
\begin{equation}\label{dico}
\II^0=\biguplus_{i=1}^6\CC_5^{fd_i}\biguplus_{j=1}^{10}\CC_3^{vd_j}\biguplus_{l=1}^{15}\CC_2^{ed_l}
\end{equation}
where vertex, edge and face axes are noted, respectively: $vd_i$, $ed_j$ and $fd_j$, the details can be found on figure \ref{dode2}. 
The vertex axes $vd_j$ are characterized by the vertices $D_j$ for $j=1\cdots 10$.
\end{listerd}

\subsection{Planar subgroups}

\subsubsection*{Cyclic subgroups}

First of all, we begin by this following lemma

\begin{lem}\label{intercy}
For every two integers $m$ and $n$ greater than $2$, and for every two axes $a$ and $b$: 
\begin{listerd}
\item If $a\neq b$ then $\CC_n^a\cap \CC_m^b=\id$
\item If $a=b$ then by noting $d:=gcd(m,n)$ we will have $\CC_n^a\cap \CC_m^b=\CC_d^a$
\end{listerd}
\end{lem}
\begin{proof}
Let $g\in \CC_n^a\cap \CC_m^b$, with $a\neq b$. $a$ and $b$ are generated by two non-collinear eigenvectors for $g$, with eigenvalue $1$. As $\det g =1$ the third eigenvalue is also one, therefore $g=e$. Thus we have the first point of the lemma. 
If now we take, for example, a common rotation of $\CC_n^0$ and $\CC_m^0$, then this rotation correspond to an angle $\theta=\frac{2l\pi}{n}=\frac{2r\pi}{m}$  with $r,l$ integers. Thus $lm=rn$ and, noting $m=dm_1$ and $n=dn_1$ we will have $lm_1=rn_1$. As $m_1$ and $n_1$ are relatively prime, we deduce that
\[
l=\alpha n_1 \text{ and then } \theta=\frac{2l\pi}{n}=\frac{2\alpha \pi}{d}\in \CC_d^0
\]
The converse inclusion is obvious, so we can conclude the lemma. 
\end{proof}

A direct application of lemma~\ref{intercy} to the intersection $\CC_n^0 \cap \CC_n^g$ leads to the result: 
\begin{lem}
For every integers $n$ and $m$, we note $d=gcd(n,m)$; we have $[\CC_n]\circledcirc [\CC_m]=\lbrace [\id],[\CC_d]\rbrace$
\end{lem}

\subsubsection*{Dihedral subgroups}

Let consider first the intersection $\Gamma=\DD_n^0 \cap \CC_m^{a}$. As $\DD_n^0=\CC_n^0 \biguplus_{l=1}^n \CC_2^{b_l}$ the following cases has to be considered: 
\begin{listerd}
\item When  $Oz=a$, the intersection $\Gamma=\CC_n^0\cap \CC_m^{a}$ and one can apply lemma~\ref{intercy}; 
\item When, for some $l$, $a=b_{l}$, $\CC_n^0\cap \CC_m^{a}=\id$ and one has to considered $\CC_2^{b}\cap \CC_m^{a}$, which equals the identity as soon as $m$ is odd. 
\end{listerd}
\begin{lem}
For every two integers $n$ and $m$, we note $d:=gcd(n,m)$ and $d_2(m):=gcd(m,2)$; then we have
\[
[\DD_n]\circledcirc [\CC_m]=\lbrace [\id],[\CC_{d_2(m)}],[\CC_d] \rbrace
\]
\end{lem}

Now consider the second kind of intersection
\begin{equation*}
\Gamma=\DD_n^0 \cap \DD_m^{g}=(\DD_n^0=\CC_n^0 \biguplus_{l=1}^n \CC_2^{b_l}) \cap( \DD_m^{g}=\CC_m^{a}\biguplus_{l=1}^m \CC_2^{gb_l})
\end{equation*}
The following cases have to be considered
\begin{listerd}
\item When  $Oz=a$, and  $Ox=gb_l0$ for some $l$ : if $d=1$ then $\Gamma=\CC_2^{b_0}$ and $\DD_m^0$ otherwise; 
\item When  $Oz=a$, and  $Ox\neq gb_l$ : $\Gamma=\CC_d^0$; 
\item When $Oz=gb_l$ for some $l$ : if $n$ is even then $\Gamma=\CC_2$ and $\id$ otherwise. Results are the same when the primary axis of $\DD_m^g$ coincides to a secondary axis of $\DD_n^0$.  
\end{listerd}
\begin{lem}
For every integers $n$ and $m$, we note $d:=gcd(n,m)$ and
\[
d_2(m):=gcd(m,2)\:; \: d_2(n):=gcd(n,2)\:; \: dz:=\begin{cases} 2 \text{ if } d=1 \\ 1 \text{ otherwise} \end{cases}
\]
Then we have $[\DD_n]\circledcirc [\DD_m]=\lbrace [\id],[\CC_{d_2(n)}],[\CC_{d_2(m)}],[\CC_{dz}] ,[\CC_d],[\DD_d]$
\end{lem}

\subsection{Clips operations on exceptional and maximum subgroups}\label{calculsbis}

Here we are concern with the subgroups $\TT^0$, $\OO^0$, $\II^0$, $\sod^0$, and $\ode^0$. For these studies, we will use results concerning their proper subgroups \cite{GSI84}. This information is summed-up in the following diagram\footnote{The arrows of figure~\ref{lattice} are to be understood in terms of partial order, that is inclusion of conjugates.} \cite{AKP12}: 

\begin{figure}[h]
\includegraphics[scale=0.75]{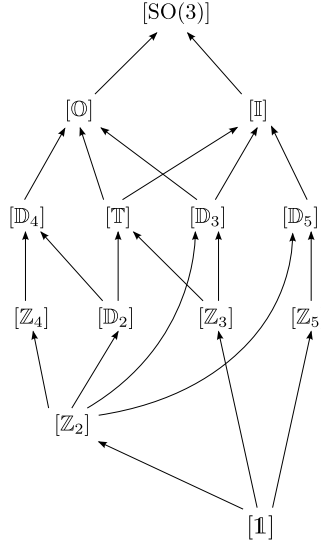}
\caption{Exceptional subgroups in the poset of closed subgroup of $\sot$}\label{lattice}
\end{figure}

\subsubsection*{Tetrahedral subgroup}

First of all we take back the decomposition~\ref{dtetra}: 
\[
\TT^0=\CC_3^{vt_1}\uplus \CC_3^{vt_2}\uplus \CC_3^{vt_3}\uplus \CC_3^{vt_4} \uplus \CC_2^{et_1}\uplus \CC_2^{et_2}\uplus \CC_2^{et_3}
\]
We begin by the study of $\TT^0 \cap \CC_n^{a}$. As a consequence of the lemma~\ref{intercy}, the primary axis of $\CC_n^{a}$ must be an edge axis or a face axis of the tetrahedron. We therefore directly obain the lemma: 
\begin{lem}
For every integer $n$, we note $d_2(n):=gcd(n,2)$ and $d_3(n):=gcd(3,n)$; then we have $
[\CC_n]\circledcirc [\TT]=\lbrace [\id],[\CC_{d_2(n)}],[\CC_{d_3(n)}] \rbrace
$
\end{lem}

Now let us  consider $\Gamma=\TT^0\cap \DD_n^g$. We will use primary axes and secondary axes of dihedral subgroup: 
\[
\DD_n^g=\CC_n^{a}\biguplus \CC_2^{b}
\]
We recall that the vertex axes of the tetrahedron are noted $vt_i$ and his edge axes are noted $et_j$. It is clear that $\Gamma$ is a subgroup of $\TT^0$. Furthermore: 
\begin{listerd}
\item As soon as $3\mid n$ if $a=vt_i$ then $\Gamma=\CC_3$ is maximal; 
\item When $2\mid n$ we can find $g$ such as $a=et_j$, then if $b=et_j$, $\Gamma=\DD_2$ or $\Gamma=\CC_2$ otherwise.
\item In any case, when we only have $b=et_j$, then $\Gamma=\CC_2$. 
\end{listerd}

Finally we can deduce the lemma:
\begin{lem} 
For every integer $n$ we note $d_2(n):=gcd(2,n)$ and $d_3(n):=gcd(3,n)$; then we have $
[\DD_n]\circledcirc [\TT]=\lbrace [\id],[\CC_2],[\CC_{d_3(n)}],[\DD_{d_2(n)}] \rbrace$
\end{lem}

Now, for the study of $\Gamma=\TT^0\cap \TT^g$ the arguments will be based on $\TT$ subgroups as well as on the axes. 
\begin{listerd}
\item First, we can find a $g$ such as all the axes are modified; in this case $\Gamma=\id$; 
\item A rotation around a face or an edge axis can be found such as only this axis is leave fixed. Then $\Gamma=\CC_2$  or $\Gamma=\CC_3$ depending on the fixed axis.
\item If we have $\Gamma\supset \DD_2$ then we can deduce that $g$ carry two edge axes onto two edge axes. After a given permutation of axes (which leave fix $\TT^0$) we can suppose that $g$ leaves fix axes $vt_1$ and $vt_2$; we then conclude that $g$ fixes also the axis $vt_3$ and then $\TT^g=\TT^0$. Thus we have here $\Gamma=\TT^0$
\end{listerd}

We deduce here the lemma:
\begin{lem} 
We have $[\TT]\circledcirc [\TT]=\lbrace [\id],[\CC_2],[\CC_3],[\TT]$ 
\end{lem}

\subsubsection*{Octahedral subgroup}

We begin by taking back the decomposition ~\ref{deccube}
\[
\OO^0=\biguplus_{i=1}^3 \CC_4^{fc_i} \biguplus_{j=1}^4 \CC_3^{vc_j} \biguplus_{l=1}^6 \CC_2^{ec_l}
\]
And, as in the case of the tetrahedron, we directly get the lemma:
\begin{lem}
For every integer $n$, we note 
\[
d_2(n)=gcd(n,2)\:; \: d_3(n)=gcd(n,3) \:; \: d_4(n)=\begin{cases} 4 \text{ if } 4\mid n \\ 1 \text{ otherwise}  \end{cases}
\]
Then we have $[\CC_n]\circledcirc [\OO]=\lbrace [\id],[\CC_{d_2(n)}],[\CC_{d_3(n)}],[\CC_{d_4(n)}]$
\end{lem}

To study the clips operation with dihedral groups, we proceed in the same way as for the tetrahedron subgroup.
The purpose is to examine axes of the cube and dihedral group. The arguments are about the same in each case. Therefore we will only detail the cases where $4\nmid n$, $3\mid n$ and $n$ is odd. 
\begin{listerd}
\item Either we have $a=vc_j$, then if  $b=fc_j$, $\Gamma=\DD_3$ or $\Gamma=\CC_3$ otherwise.
\item Or $a=fc_j=0$, then if $b=fc_j$, $\Gamma=\CC_2$ or $\Gamma=\id$ otherwise.
\item Else $a=ec_j$, then if $b=ec_j$ or $b=fc_j$, $\Gamma=\CC_2$.
\end{listerd}

All these arguments leads us to the lemma:
\begin{lem}
For every integer $n$, we note 
\[
d_2(n):=gcd(n,2)\:; \: d_3(n):=gcd(n,3) \:; \: d_4(n):=\begin{cases} 4 \text{ if } 4\mid n \\ 1 \text{ otherwise}  \end{cases}
\]
Then we have $
[\DD_n]\circledcirc [\OO]=\lbrace [\id],[\CC_2],[\CC_{d_3(n)}],[\CC_{d_4(n)}],[\DD_{d_2(n)}],[\DD_{d_3(n)}],[\DD_{d_4(n)}]\rbrace$
\end{lem}

Now, we take $\Gamma=\OO^0 \cap \TT^g$ and one can observe that $\Gamma$ is necessarily a common subgroup of $\OO^0$ and $\TT^g$, then its class must contain (in the sense of the partial order) $[\id]$, $[\CC_2]$, $[\CC_3]$, $[\DD_2]$ or $[\TT]$. After that: 
\begin{listerd}
\item There exists a rotation $g$ around an edge axis of $\TT^0$ (that is a common face axis of the cube) such that only this axis is fixed; and then $\Gamma=\CC_2$;
\item There exists a rotation $g$ around a vertex axis of $\TT^0$ (that is a common vertex axis of the cube) such that only this axis is fixed; and then $\Gamma=\CC_3$;
\item As soon as $\Gamma\supset \DD_2$, as in tetrahedral case, we  necessarily have $\Gamma=\TT^0$. 
\end{listerd}

We conclude here the lemma:
\begin{lem}
We have $[\TT]\circledcirc [\OO]=\lbrace [\id],[\CC_2],[\CC_3],[\TT] \rbrace$
\end{lem}

For the study of $\Gamma=\OO^0\cap \OO^g$ we will also use arguments based on subgroups. Some results are nevertheless more subtle:  
\begin{listerd}
\item First, there exists a rotation $g$ that fix only one edge axis, and in that case $\Gamma=\CC_2$; 
\item Then there exists a rotation that leaves fix only one vertex axis, and in that case $\Gamma=\CC_3$; 
\item There exists also a rotation that leaves fix only one face axis, and no other axis is fixed. See for example figure~\ref{cubec4}: in that case $\Gamma=\CC_4$;
\item We can also find a rotation that leaves fix a face axis and which bring an edge axis onto a face axis. Indeed when we take $g=\QQ\left( \ii;\frac{\pi}{4}\right)$ we obtain $\Gamma=\CC_4^{\ii}\uplus \CC_2^{\kk}=\DD_4$ (c.f. figure~\ref{cubed4});
\item If we take  $g=\QQ\left( \kk;\frac{\pi}{4}\right) \circ \QQ\left( \ii;\frac{\pi}{4}\right)$ we directly obtain $\Gamma=\DD_2$. We can exactly compute that $gfc_3=ec_6$, $gec_1=fc_1$ and  $gec_2=ec_5$ and no other axes corresponds (c.f. figure~\ref{cubeklein}). 
\item If we take $g=\QQ(vc_1,\pi)$ we will find $\Gamma=\DD_3$ with $vc_1$ as primary axis and $ec5$ a secondary axis.
\item If $\Gamma\supset \TT$ then, necessarily, $g$ leaves fix the three edge axes of the tetrahedron, and then $g$ will fix the cube $\mathcal{C}_0$; thus $\Gamma=\OO^0$. 
\end{listerd}

\begin{figure}
    \centering
    \subfigure[$\OO^0\cap \OO^g=\DD_2$]{\includegraphics[scale=0.75]{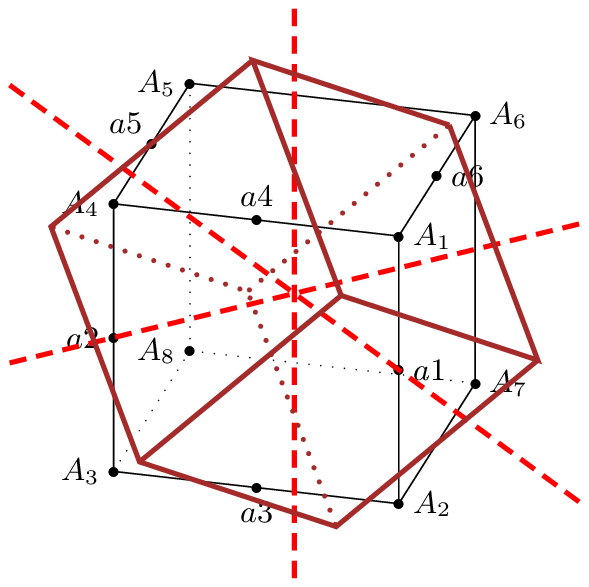}\label{cubeklein}}
    \subfigure[$\OO^0\cap \OO^g=\CC_4$]{\includegraphics[scale=0.9]{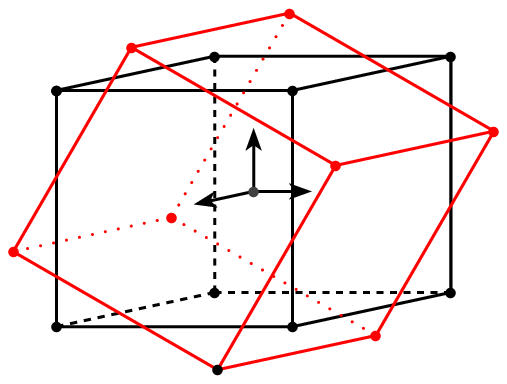}\label{cubec4}}
    \subfigure[$\OO^0\cap \OO^g=\DD_4$]{\includegraphics[scale=0.9]{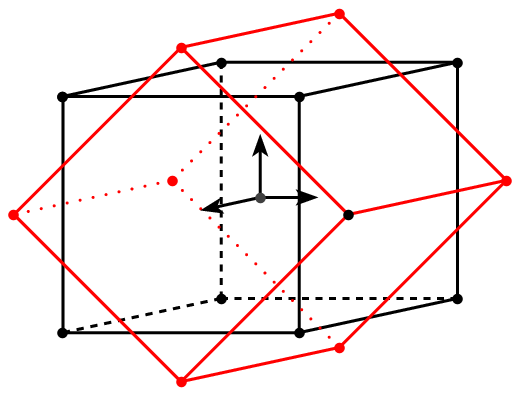}\label{cubed4}}
\end{figure}

Finally we get the lemma:
\begin{lem}
We have $[\OO]\circledcirc [\OO]=\lbrace [\id],[\CC_2],[\DD_2],[\CC_3],[\DD_3],[\CC_4],[\DD_4],[\OO] \rbrace$
\end{lem}

\subsubsection*{Icosahedral subgroup}

We take the decomposition~\ref{dico}
\[
\II^0=\biguplus_{i=1}^6\CC_5^{fd_i}\biguplus_{j=1}^{10}\CC_3^{vd_j}\biguplus_{l=1}^{15}\CC_2^{ed_l}
\]
And, as in the previous situations, we directly get the lemma:
\begin{lem}
For every integer $n$, we note 
\[
d_2:=gcd(n,2)\:; \: d_3:=gcd(n,3)\:; \: d_5:=gcd(n,5)
\]
Then we have $[\CC_n]\circledcirc [\II]=\lbrace [\id],[\CC_{d_2}],[\CC_{d_3}],[\CC_{d_5}] \rbrace$
\end{lem}

Now, for the study of $\II^0\cap \DD_n^g$ we again use arguments about axes:
\begin{listerd}
\item If $a=ft_j$ or $a=vt_j$, then $\Gamma\in\{\CC_{d_3}, \CC_{d_5},\DD_{d_3}\DD_{d_5}\}$ ;
\item If $a=et_j$ then $\Gamma\in\{\CC_{d_2}, \DD_{d_2}\}$.
\end{listerd}
When we argue on the secondary axis of $\DD_n^g$, we see that we can always have $\CC_2$.

Finally we get the lemma:  
\begin{lem}
For every integer $n$, we note 
\[
d_2:=gcd(n,2)\:; \: d_3:=gcd(n,3)\:; \: d_5:=gcd(n,5)
\]
Then we have $
[\DD_n]\circledcirc [\II]=\lbrace [\id],[\CC_2],[\CC_{d_3}],[\CC_{d_5}],[\DD_{d_2}],[\DD_{d_3}],[\DD_{d_5}]  \rbrace\rbrace$
\end{lem}

For the intersection $\II^0 \cap \TT^g$ it is clear, because of the inclusion $\TT^0 \subset \II^0$ that we can obtain all the classes of $[\TT]\circledcirc [\TT]$. If now this intersection contains a subgroup $\DD_2$, we will necessarily have 
\[
\DD_2=\CC_2^{get_1}\uplus \CC_2^{get_3}\uplus \CC_2^{get_3}
\]
where $get_i$ are the three edge axis of the tetrahedron $g\mathcal{T}_0$. These three axes will then have to correspond to three perpendicular axes of the dodecahedron. After permutation of the axes, which leaves fix the dodecahedron, we can suppose that these three axes are generated by the three vectors of the basis. But, then, the vertex axes of the tetrahedron will correspond to vertex axes of the embedded cube in the dodecahedron. We can then deduce that the intersection will be the whole $\TT$ subgroup.  

Now we have to study $\Gamma=\II^0 \cap \OO^g$. For that, we refer to the common subgroups of $[\OO]$ and $[\II]$. Such subgroups can clearly be taken from the poset~\ref{lattice}. First, it is clear that, when the cube related to $\OO^g$ is the embedded cube in the dodecahedron, we will have $\Gamma=\TT$. 

We also can find a rotation $g$ such as $\Gamma$ contains $\DD_3$: indeed, $g$ has to bring the vertex axis of the cube $vc_1$ onto the vertex axis of the dodecahedron $vd_5$ and the edge axis of the cube $ec_5$ onto the edge axis of the dodecahedron $ed_7$. With maximality argument, we can deduce that $\Gamma=\DD_3$. We now have to examine the case of $\DD_2$, $\CC_3$ and $\CC_2$: 
\begin{listerd}
\item When $\Gamma \supset \DD_2$, then, after permutation of axis of the dodecahedron, we can suppose that $g$ leaves fix the three axis of the basis vector. But these three axis are two order rotations of the dodecahedron. Thus, $g$ will fix the cube $\mathcal{C}_0$ and we can deduce that $\OO^g=\OO^0$. We then have $\Gamma=\TT$.
\item We can find a rotation $g$ around a vertex axis, for example $vd_5$ such that $\Gamma=\CC_3$; 
\item As above, we can find a rotation around an edge axis, such that $\Gamma=\CC_2$. 
\end{listerd}

We finally conclude to the formula: 
\[
[\OO]\circledcirc [\II]=\lbrace [\id],[\CC_2],[\CC_3],[\DD_3],[\TT] \rbrace
\] 
For the intersection $\Gamma=\II^0\cap \II^g$ we will have to study the case of classes $[\TT]$, $[\DD_3]$, $[\DD_5]$, $[\DD_2]$, $[\CC_3]$, $[\CC_5]$ et $[\CC_2]$: 
\begin{listerd}
\item When $\Gamma\supset \TT$ or $\Gamma\supset \DD_2$, it then contains all the three two order rotations around each base axes, which will be three edge axes of the dodecahedron. Thus we can deduce that $g$, after permutation of these axes, leaves fix three perpendicular axes, and then $g$ leaves fix $\II^0$; finally $\Gamma=\II$; 
\item There exist a rotation $g$ around an edge axis so that $\Gamma=\CC_2$. The same argument leads us to $\CC_3$ and $\CC_5$. 
\item If we take $g$ to be the two order rotation around the axis $vd_3$, we can compute that this rotation only leaves fix the axes $vd_3$, $ed_6$, $ed_8$ and $ed_{15}$, and then $\Gamma=\DD_3$. 
\item If we take $g$ to be the two order rotation around the face axis $fd_1$ we can also compute that it only leaves fix the axes $fd_1$, $ed_7$, $ed_{11}$, $ed_{12}$ and $ed_{14}$, thus $\Gamma=\DD_5$. 
\end{listerd}
%

\bibliographystyle{plain}
\bibliography{elasticity1} 
\nocite{*}
\end{document}